\documentclass[draftclsnofoot,onecolumn,12pt,romanappendics]{IEEEtran}

\usepackage{amsfonts}
\usepackage{slashbox}
\usepackage{multirow}
\usepackage{graphicx}
\usepackage{amsfonts}
\usepackage{amsmath,epsfig}
\usepackage{mathbbold}
\usepackage{stfloats}
\usepackage{url}
\usepackage{bm}
\usepackage{cite}
\usepackage{latexsym}
\usepackage{epsfig}
\usepackage{graphics}
\usepackage{mathrsfs}
\usepackage{algorithmic}
\usepackage{algorithm}
\usepackage{subfigure}
\usepackage[all]{xy}
\usepackage{pifont}
\usepackage{bbding}
\usepackage{stmaryrd}
\usepackage{amsfonts}
\usepackage{epic}
\usepackage{latexsym}
\usepackage{epstopdf}
\usepackage{epic}
\usepackage{xcolor}
\usepackage{mathrsfs}
\usepackage{pifont}
\usepackage{epsfig}
\usepackage{mathbbold}
\usepackage{subfigure}

\usepackage{enumerate}
\usepackage{setspace}
\usepackage{latexsym}
\usepackage{booktabs}
\usepackage{color}
\usepackage{cases}
\usepackage{amsthm}
\usepackage{amssymb}

\usepackage{hyperref}

\newtheorem{theorem}{\underline{Theorem}}
\newtheorem{proposition}{\underline{Proposition}}
\newtheorem{lemma}{\underline{Lemma}}
\newtheorem{remark}{\underline{Remark}}

\newcommand{\mv}[1]{\mbox{\boldmath{$ #1 $}}}
\newcommand{\Q}{\mathbf Q}

\newcommand{\C}{\mathcal C}

\newcommand{\pow}{\mathbf P}

\newcommand{\I}{\rm{I}}
\newcommand{\F}{\rm{F}}
\newcommand{\Aaa}{\rm{A}}

\newcommand{\X}{\mathcal{X}}
\newcommand{\Aa}{\mathcal{A}}
\newcommand{\Bb}{\mathcal{B}}
\newcommand{\T}{\mathcal{T}}

\begin{document}
\title{Capacity Characterization of UAV-Enabled Two-User Broadcast Channel}
\author{\IEEEauthorblockN{Qingqing Wu,  Jie Xu,  \emph{Member, IEEE}, and Rui Zhang, \emph{Fellow, IEEE}
\thanks{ Q. Wu and R. Zhang are with the Department of Electrical and Computer Engineering, National University of Singapore (email:\{elewuqq, elezhang\}@nus.edu.sg). J. Xu is with the School of Information Engineering, Guangdong University of Technology (email: jiexu@gdut.edu.cn). }}  }
%
\maketitle
\vspace{-1.5cm}
\begin{abstract}
Unmanned aerial vehicles (UAVs) have recently gained growing popularity in wireless communications owing to their many advantages such as swift and cost-effective deployment, line-of-sight (LoS) aerial-to-ground link, and controllable mobility in three-dimensional (3D) space. Although prior works have exploited the UAV's mobility to enhance the wireless communication performance under different setups, the fundamental capacity limits of UAV-enabled/aided  multiuser communication systems have not  yet been characterized.
To fill this gap, we consider in this paper a UAV-enabled two-user broadcast channel (BC), where a UAV flying at a constant  altitude is deployed to send independent  information to two users at different fixed locations on the ground.
We aim to characterize the capacity region of this new type of BC over a given UAV flight duration,  by jointly optimizing the UAV's trajectory and transmit power/rate  allocations over time, subject to the UAV's maximum speed and maximum transmit power constraints.
First, to draw essential  insights, we consider two special cases with asymptotically large/low UAV  flight duration/speed, respectively. For the former case, it is shown that a simple hover-fly-hover (HFH) UAV trajectory with time division multiple access (TDMA) based  orthogonal multiuser transmission is capacity-achieving;  while in the latter case,  the UAV should hover at a fixed location that is nearer to the user with larger achievable rate and in general  superposition coding (SC) based non-orthogonal transmission with interference cancellation at the receiver of the nearer user is required.
Next, we consider the general case with finite UAV speed and flight duration. We show that the optimal UAV trajectory should follow a general HFH structure, i.e., the UAV successively hovers at a pair of initial and final locations above the line segment of the two users each with a certain amount of time and flies unidirectionally between them at the maximum speed, and SC is  generally needed.
Furthermore, when TDMA-based transmission is considered for low-complexity implementation, we show  that the optimal UAV trajectory still follows a HFH structure, but the hovering locations can only be those above the two users. Finally, simulation results are provided to verify our analysis, which also reveal useful guidelines  to the practical design of UAV trajectory and communication jointly. 

\end{abstract}


\section{Introduction}
On November 8, 2017, ``Drone Integration Pilot Program''  \cite{USprogramm_UAV} was launched under a presidential memorandum from the White House, which aimed at further exploring expanded use of  unmanned aerial vehicles (UAVs)  including beyond-visual-line-of-sight flights, night-time operations, flights over people, etc. \cite{whitehouse_UAV}.
In fact, the past several years have witnessed an unprecedented growth on the use of UAVs in a wide range of civilian and defense applications such as search and rescue, aerial  filming and inspection, cargo/packet delivery, precise agriculture, etc.  \cite{drone_num}.
In particular, there has been a fast-growing  interest in utilizing UAVs as aerial communication platforms to help enhance  the performance and/or extend the coverage of existing wireless networks on the ground \cite{zeng2016wireless,bor2016new}.
For example, UAVs such as drones, helikites, and balloons, could be deployed as aerial  base stations (BSs) and/or relays to enable/assist the terrestrial communications.
 UAV-enabled/aided  wireless  communications possess many appealing advantages such as swift and cost-effective deployment, line-of-sight (LoS)
aerial-to-ground link, and controllable mobility in three-dimensional (3D) space,  thus highly promising for numerous use cases in wireless communications including ground BS traffic offloading, mobile relaying and edge computing, information/energy broadcasting and data collection for Internet-of-Things (IoT) devices, fast network recovery after natural disasters,  etc. \cite{3gpp_UAV,qualcom_UAVreport,zhang2017spectrum,xu2017uav2,yang2017proactive,yang2017energy,zhang2016fundamental}. For example, Facebook has even ambitiously claimed that ``Building drones is more feasible than covering the world with ground signal towers'' \cite{fb_UAV}.   By leveraging the aerial BSs along with terrestrial and satellite communications, Europe has established an industry-driven project called  ``ABSOLUTE'' with the ultimate goal of enhancing the ground network capacity to many folds, especially for public safety in emergency situations \cite{absolute}.  At present, there are two major ways to practically implement aerial BSs/relays by using tethered and untethered UAVs, respectively, which are further explained as follows.

A tethered UAV literally means that the UAV is connected by a cable/wire with a ground control platform (e.g., a custom-built trailer). Although it may sound ironic for a UAV to be on a tethering cable, this practice is very common due to many advantages including  stable power supply and hence unlimited endurance, more affordable payload (e.g., more antennas),  ultra-high speed backhaul with secured data transmission (e.g.,  real-time high-definition video), robustness to wind, etc.
 All these evident benefits have triggered a great interest in testing tethered UAV BSs, such as
Facebook's ``Tether-Tenna'',  AT\&T's ``flying cell-on-wings (COWs)'', and Everything-Everywhere's (UK's largest mobile network operator, EE) ``Air Masts''.
However, such a tethering feature also limits the operations of UAVs to taking off, hovering, and landing only, thus rendering the wireless networks employing tethered UAV BSs like ``hovering  cells'' over the air. As a result, the research efforts in this paradigm  mainly focus on the UAV deployment/placement optimization in a given target area to meet the ground data  traffic demand
 \cite{al2014optimal,lyu2016placement,bor2016efficient, mozaffari2016efficient,yang2017proactive,alzenad20173d}. In particular,  \cite{al2014optimal} provides an analytical approach to optimize the altitude of a UAV for providing the maximum coverage for ground users (GUs). Alternatively, by fixing the altitude, the horizontal positions of UAVs are optimized  in \cite{lyu2016placement} to minimize the number of required UAVs  to cover a given set of GUs. 
In contrast to tethered UAVs, the untethered UAVs generally rely on on-board battery and/or solar energy harvesting for power supply  \cite{wu2016overview}
and wireless links for data backhaul. Although untethered UAVs in general have smaller payload and limited endurance/backhaul rate as compared to their tethered counterparts, they have fully controllable mobility in 3D space which can be exploited for communication performance enhancement.
 First,  an untethered UAV  not only can hover above a fixed ground location like tethered UAVs,   but also can fly freely over a wide ground area to significantly extend the communication coverage.
  Furthermore, the free-flying feature enables the UAV BS to timely adjust its position according to the dynamic distributions of the GUs, and even follow closely some specific GUs, to achieve a new ``user-centric and cell-free'' communication service. 
 In fact, the reduced UAV-GU distance not only decreases the signal attenuation but also increases the probability of the LoS link between them, which is particularly crucial for high-rate communication.
  As a result, untethered UAV  BSs/relays have been envisioned to be a revolutionizing technology for future wireless communication systems \cite{fotouhi2017dronecells} and preliminary industry prototypes have been built and tested including e.g. Facebook's  Aquila and Nokia's flying-cell (F-Cell).
  This has also inspired a proliferation of studies recently on the new research paradigm of jointly optimizing the UAV trajectory design and communication resource allocation, for e.g. mobile relaying channel \cite{zeng2016throughput,li2016energy}, multiple access channel (MAC) and broadcast channel (BC) \cite{wu2017joint,lyu2016cyclical,JR:wu2017_ofdm}, interference channel (IFC)\cite{JR:wu2017joint}, and  wiretap channel \cite{guangchi2016_UAV}.  In particular,  as shown in \cite{lyu2016cyclical} and \cite{JR:wu2017_ofdm}, significant communication throughput gains can be achieved by mobile UAVs over static UAVs/fixed terrestrial BSs by exploiting the new design degree of freedom of UAV trajectory optimization, especially for delay-tolerant applications.  In \cite{JR:wu2017joint},  a joint UAV trajectory, user association, and power control scheme is proposed for cooperative  multi-UAV enabled wireless networks.  A multi-objective path
planning (MOPP) framework is proposed in  \cite{yin2017offline} to explore a suitable path for a UAV operating in a dynamic urban environment.
 In \cite{mozaffari2017mobile},  an efficient mobile architecture is proposed  for uplink data collection application in IoT networks.
 To optimize the wireless system performance by exploiting UAV-enabled BSs,  assorted  UAV trajectory designs have been proposed in the literature \cite{li2016energy,wu2017joint,JR:wu2017joint,lyu2016cyclical,JR:wu2017_ofdm,jeong2016mobile,zeng2016throughput,mozaffari2016unmanned,mozaffari2017mobile,wu2017ofdma}, based on optimization  techniques such as successive convex optimization (SCA) and solutions for  Travelling Salesman Problem (TSP) as well as its various extensions.
 However, all these works either assume time division multiple access (TDMA) \cite{li2016energy,wu2017joint,JR:wu2017joint,lyu2016cyclical,mozaffari2016unmanned} or frequency division multiple access (FDMA)  \cite{JR:wu2017_ofdm,jeong2016mobile,zeng2016throughput,mozaffari2017mobile,wu2017ofdma} to simplify the multiuser communication design, which, however, are in general  suboptimal from an information-theoretic perspective. As a result, the fundamental capacity limits of UAV-enabled multiuser communication systems still  remain largely unknown, which thus motivates this work. 


\begin{figure}[!t]
\centering
\includegraphics[width=0.4\textwidth]{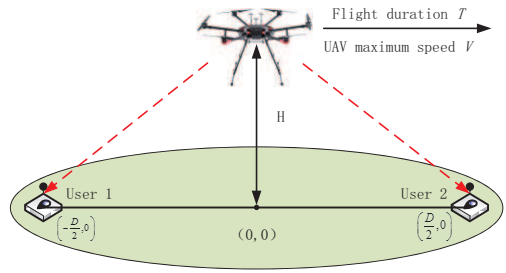}
\caption{A UAV-enabled broadcast channel (BC) with two GUs at fixed locations. } \label{system:model}\vspace{-0.5cm}
\end{figure}

In this paper, we aim to characterize the capacity region of a UAV-enabled BC and reveal the capacity-optimal joint UAV trajectory and communication design.
As an initial study, we consider the simplified setup with two GUs, as shown in Fig. \ref{system:model}. Specifically, it is assumed that a UAV with the maximum speed $V$ meter/second (m/s) flies at a constant altitude of $H$ m to serve two GUs at fixed locations with a distance of $D$ m. We consider the communication within  a given UAV flight duration of $T$ s. Note that if $VT\ll D$, the considered system simplifies to e.g. a tethered UAV BS, which can be placed above a fixed ground location.  On one hand, we assume that $H$ is sufficiently large such that the channels from the UAV to both GUs are dominated by the LoS link, according to the recently conducted experimental results by Qualcomm \cite{qualcom_UAVreport}. On the other hand, we assume that $D$ is sufficiently large and comparable to $H$ so that the UAV's horizontal position can have a non-negligible impact on the UAV-GU channel strength.\footnote{Otherwise, if $D\ll H$, the UAV trajectory design becomes trivial and it should simply stay above any point along the line between the two GUs and their channels can be regarded as constant irrespective of $D$. }
As a result, an effective time-varying BC can be generally established between the UAV and the two GUs  as the UAV moves horizontally above them. Given the UAV trajectory, the UAV-GU BC resembles the conventional fading BC with a terrestrial BS \cite{tse2005fundamentals}.  However, their fundamental difference lies in  that the UAV trajectory and hence its induced  time-varying channel are controllable and thus can be proactively designed to maximize the capacity of the BC, while this is impossible for conventional fading channels due to the randomness in the propagation environment.
As such, the joint UAV trajectory and communication design can exploit this  additional  degree of freedom to enlarge the capacity region compared to the conventional BC with a static BS on the ground.

To this end, we characterize the  capacity region of this new UAV-enabled BC  over a given UAV flight duration $T$, by jointly optimizing the UAV's trajectory and transmit power/rate
allocations over time, subject to the practical UAV's maximum speed and transmit power constraints.
Specifically, we adopt the \emph{rate-profile} approach as in \cite{zhang2010cooperative} to maximize the sum rate of the two GUs under their different rate ratios, which leads to a complete characterization of all the achievable rate-pairs for the two GUs on the so-called \emph{Pareto} boundary of the capacity region. However, such a joint optimization problem is shown to be non-convex and difficult to solve in general.  Nevertheless,  we obtain the optimal solution of the considered problem by exploiting its particular structure and applying tools from convex optimization. The main results of this paper are summarized as follows:
\begin{itemize}
  \item First, to draw essential insights, we consider two special cases with asymptotically large UAV  flight duration, i.e.,  $T\rightarrow \infty$ (or equivalently $VT\gg D$) and asymptotically low UAV speed, i.e.,  $V\rightarrow 0$ (or equivalently $VT \ll D$), respectively. We introduce a simple and practical UAV trajectory called hover-fly-hover
(HFH), where the UAV successively hovers at a pair of initial and final locations above the line segment connecting the two GUs each with a certain amount of time, and flies unidirectionally between them at its maximum speed. Then, for the case of $T\to \infty$,  we show that the HFH trajectory with hovering locations above the two GUs together  with the TDMA  based  orthogonal multiuser transmission is capacity-achieving.
 In contrast, for  the case of $V\rightarrow 0$, it is shown that the UAV should hover at a fixed location that is nearer to the GU with larger achievable rate and in general superposition coding (SC) based non-orthogonal transmission with interference cancellation at the receiver of the nearer GU is required.
Furthermore,  it is shown that in general there exists a significant capacity gap between the above two cases, which demonstrates the potential  of exploiting the UAV's trajectory design and  motivates the study for the general case with  finite  UAV maximum speed and flight duration.
  \item Next, for the case of finite UAV speed and flight duration, we prove that the proposed HFH trajectory is also optimal while SC is generally required to achieve the capacity. In addition, the initial and final hovering locations need to be properly selected from  the points above the line segment between the two GUs to achieve the capacity region Pareto boundary. It is also observed that by increasing the UAV maximum speed and/or flight duration, the capacity region is effectively enlarged, especially for the low signal-to-noise ratio (SNR) case.
  To gain more insights, we further analyze the high SNR case and  it is shown that the HFH trajectory reduces to a static point above one of the two GUs. This result implies that dynamic UAV movement is less effective for capacity enhancement as SNR increases. 
  \item Last, for the sake of  comparison, we further characterize the achievable rate region of the UAV-enabled two-user BC with the TDMA-based (instead of the optimal SC-based) transmission with finite UAV speed and flight duration. It is shown that the optimal UAV trajectory still follows the HFH structure as in the capacity-achieving case with SC-based transmission,  while the difference  lies in that the hovering locations can only be those above the two GUs in the TDMA case. It is also revealed  that the capacity gain of the optimal SC-based transmission over the suboptimal TDMA-based transmission
decreases as the UAV maximum speed and/or flight duration increases.
\end{itemize}

The rest of this paper is organized as follows. Section II introduces the system model  and presents the problem formulation for capacity region characterization. In Sections III-V, we study the capacity region for two special cases and the general case, respectively. Section \ref{Section:TDMA} addresses the case with TDMA-based transmission.  Finally, we conclude the paper in Section VI.

\emph{Notations:} In this paper, scalars are denoted by italic letters, vectors and matrices are denoted by bold-face lower-case and upper-case letters, respectively. $\mathbb{R}^{M\times 1}$ denotes the space of $M$-dimensional real-valued vectors. For a vector $\mathbf{a}$, $\|\mathbf{a}\|$ represents its Euclidean norm and $\mathbf{a}^T$ denotes its transpose. For a time-dependent function $\mathbf{x}(t)$,  $\dot{\mathbf{x}}(t)$ denotes the derivative with respect to time $t$. For a set $\Aa$, $|\Aa|$ denotes its cardinality, ${\rm{int}}( \Aa)$ and $\partial \Aa$ represent the  interior and boundary of a set $\Aa$, $\rm{Conv}(\Aa)$ represents the convex hull of a set $\Aa$, which is the set of all the convex combinations of the points in $\Aa$, i.e.,
${\rm{Conv}}(\Aa)=\big\{ \sum^{|\Aa|}_{n=1}\alpha_nc_n: \forall\, \alpha_n\geq 0, \sum_{n=1}^{|\Aa|}\alpha_n=1\big\}$. For two sets $\Aa$ and $\Bb$, $ \Aa  \backslash  \Bb$ is the set of all elements in  $\Aa$ excluding those in $\Bb$.  Notation $\mathbf{a} \preceq \mathbf{b}$ indicates that vector $\mathbf{a}$ is  element-wisely less than or equal to vector $\mathbf{b}$.
\vspace{-0.1cm}
\section{System Model and Problem Formulation}

As shown in Fig. \ref{system:model}, we consider a UAV-enabled BC with one UAV transmitting independent information to two GUs at fixed locations. Without loss of generality, we consider a two-dimensional (2D) Cartesian coordinate system. Let the location of each GU $k \in \{1,2\}$ be denoted by $(x_k,0)$, where $x_1 =-D/2$ m and $x_2 =D/2$ m with $D>0$ denoting their distance.  The UAV is assumed to fly at a constant altitude of $H$ m. In practice, the value of $H$ is set based on regulations on the minimum UAV height as well as the communication system requirement.
We focus on a particular UAV flight duration of $T$ s and denote the UAV's time-varying location at time instant $t \in \mathcal T \triangleq [0,T]$ by $(x(t), H)$. 
The system bandwidth is denoted by $B$ in Hertz (Hz) and hence the symbol period is $T_s=1/B$ s. We assume that $TB$ is sufficiently large such that the UAV can adopt the Gaussian signaling with a sufficiently long symbol block length to achieve the channel capacity. It is also assumed that the UAV's location change within a symbol period is negligible compared to the altitude $H$, i.e., $VT_s\ll H$, where $V \ge 0$ denotes the UAV maximum speed  in m/s.
Thus,  the UAV-GU channel is assumed to be constant within each symbol interval. Mathematically,  we express the UAV speed constraint as  \cite{JR:wu2017joint,JR:wu2017_ofdm},
\begin{align}\label{UAVspeed}
|\dot{x}(t)| &\le V, \forall\, t\in \mathcal T.
\end{align}

For the purpose of exposition, we consider the free-space path loss model for the air-to-ground wireless communication channels from the UAV to the two GUs, as justified in Section I. 
It is assumed that the Doppler effect induced by the UAV mobility can be perfectly compensated at the user receivers \cite{zeng2016wireless,zhang2017cellular}.  As a result,  the channel power gain from the UAV to each GU $k$ at time instant $t$ is modeled as
\begin{align}\label{eq:channelgain}
\tilde{h}_k(x(t)) = \frac{\gamma_0}{(x(t) - x_k)^2+H^2},
\end{align}
where  $\gamma_0$ denotes the channel power gain at the reference distance $d_0=1$ m. 

At time instant $t\in\mathcal T$, let $s_1(t)$ and $s_2(t)$ denote the UAV's transmitted information-bearing  symbols  for GUs 1 and 2, respectively. Accordingly, the received signal at  GU $k$ is expressed as
\begin{align}\label{eq:033}
y_k(t) = \sqrt{\tilde{h}_k(x(t))} \big(s_1(t) + s_2(t)\big)+ n_k(t),  k\in\{1,2\},
\end{align}
where $n_k(t)$ denotes the additive white Gaussian noise (AWGN) at the receiver of  GU $k$. For simplicity, the noise power is assumed to be equal for the two GUs, denoted by $\sigma^2$.
With given $x(t)$, the signal model in \eqref{eq:033} resembles a conventional fading BC consisting of one transmitter (the UAV) and two receivers (GUs) \cite{tse2005fundamentals}. In order to achieve the capacity region of this channel, the UAV transmitter should employ Gaussian signaling by setting $s_k(t)$'s as independent  circularly symmetric complex Gaussian (CSCG) random variables with zero mean and variances $p_k(t) = \mathbb{E}(|s_k(t)|^2), k\in\{1,2\}$. Suppose that at each time instant $t$, the UAV is subject to a maximum transmit power constraint $\bar P$, similarly as assumed in  \cite{tse1998multiaccess,li2001capacity}, i.e.,
\begin{align}\label{UAVpower}
p_1(t) + p_2(t) \le \bar P, \forall t\in\mathcal T.
\end{align}


In this paper, we are interested in characterizing the capacity region of the UAV-enabled two-user BC, which consists of all the achievable average rate-pairs for the two GUs  over the duration $T$, subject to the UAV's maximum speed constraint in \eqref{UAVspeed} and maximum power  constraint in \eqref{UAVpower}.
  For given UAV trajectory $\Q\triangleq\{ x(t), t\in \T\}$ and power allocation $\pow \triangleq  \{p_{k}(t), t\in \T, k\in \{1,2\}\}$, let $\C(\Q, \pow)$ denote the set of all achievable average rate-pairs ($r_1, r_2)$ in bits per second per Hertz (bps/Hz) for the two GUs, respectively, which need to satisfy the following inequalities \cite{tse1998multiaccess,jindal2004duality}: 

\begin{align}
r_1 \le & \frac{1}{T}\int_{0}^T \log_2\big(1+ p_1(t)h_1(x(t))\big) {\rm{d}}t,  \label{MAC:ineq1}\\
r_2  \le & \frac{1}{T}\int_{0}^T\log_2\big(1+ p_2(t)h_2(x(t))\big){\rm{d}}t, \label{MAC:ineq2}\\
r_1 + r_2 \le& \frac{1}{T}\int_{0}^T\log_2\big(1+ p_1(t)h_1(x(t)) +p_2(t)h_2(x(t))\big){\rm{d}}t, \label{MAC:ineq3}
\end{align}
where $h_k(x(t)) = \tilde{h}_k(x(t))/\sigma^2, k\in\{1,2\}$.
 Denote by $\X_1$ and $\X_2$ the feasible sets of  $\Q$ and $\pow$ specified by the UAV's speed constraint \eqref{UAVspeed} and maximum power  constraint \eqref{UAVpower}, respectively.
Then, the capacity region of the UAV-enabled two-user BC  is defined as
\begin{align}\label{capacityregion}
\C (V,T,\bar P) = \bigcup_{\Q \in \X_1,\pow \in \X_2}\C (\Q, \pow).
\end{align}

Our objective is to characterize the Pareto boundary (or the upper-right boundary) of the capacity region $\C (V,T,\bar P)$ by jointly optimizing the UAV trajectory $\Q$ and power allocation $\pow$. 
The Pareto boundary consists of all the achievable average rate-pairs at each of which it is impossible to improve the average rate of one  GU without simultaneously decreasing that of the other  GU. Since it remains unknown yet whether the capacity region  $\C (V,T,\bar P)$  is a convex set or not, we apply the  rate-profile technique in \cite{zhang2010cooperative} that ensures a complete characterization of the capacity region, even if it is non-convex.\footnote{Another commonly adopted approach for characterizing the Pareto boundary is to maximize the weighted sum of the average rates of the two GUs. Although this approach is effective  in characterizing Pareto boundary points for convex capacity regions, it may fail to obtain all the boundary points when the capacity region is a non-convex set \cite{Boyd}.}  Specifically, let $\mv\alpha = (\alpha_1,\alpha_2)$ denote a rate-profile vector which specifies the rate allocation between the two GUs with $\alpha_k\geq 0, k\in\{1,2\}$,  and $\alpha_1+\alpha_2=1$. Here, a larger value of $\alpha_k $ indicates that  GU $k$ has a higher priority in information transmission to achieve a larger average rate.
Then,  the characterization of each Pareto-boundary point corresponds to solving the following problem,
\begin{align}
\text{(P1)}: ~~\max_{\{r,r_1,r_2,\Q, \pow\}} & r \\
\mathrm{s.t.}~~~~& r_k \ge \alpha_k r, \forall\, k\in\{1,2\},\label{eq:9}\\
&(r_1,r_2) \in \C(\Q, \pow),\label{P1:11}\\
&p_1(t) + p_2(t) \le \bar P, \forall\,t\in \T,\label{P1:12}\\
&|\dot{x}(t)| \le V, \forall\, t\in \T, \label{P1:speed} \\
&p_k(t) \geq 0, \forall\, t\in \T,  k\in\{1,2\},\label{P1:14}
\end{align}
where $r$ denotes the achievable sum average rate of the two GUs. Problem (P1) is challenging to be solved optimally due to the following two reasons. First,  the constraints in (11) are non-convex, as the rate functions in $\C(\Q, \pow )$ (i.e., the right-hand-sides (RHSs) of \eqref{MAC:ineq1}, \eqref{MAC:ineq2}, and \eqref{MAC:ineq3}) are non-concave with respect to $\Q$. Second, problem (P1) involves an infinite number of optimization variables (e.g., $x(t)$'s over continuous time $t$). As a result, (P1) is a highly non-convex optimization problem and in general, there is no standard method to solve such problem efficiently.  For notational convenience,  the optimal UAV trajectory  for problem (P1) under any given $(\alpha_1,\alpha_2)$ is denoted by $\{x^*(t)\}$, and the optimal power allocation is denoted by $\{p_k^*(t), k\in\{1,2\}\}$. Furthermore, the corresponding rate-pair achieved by the above optimal UAV trajectory and power allocation is denoted by $(r^*_1, r^*_2)$, which is on the (Pareto) boundary of the capacity region $\C (V,T,\bar P)$.
\vspace{-0.25cm}
\subsection{ Capacity Region Properties and HFH  Trajectory}
Before explicitly characterizing the capacity region $\C (V,T,\bar P)$,  we first provide some interesting properties of this region, which can be used to simplify the optimization of the UAV trajectory in (P1) later. 
\begin{lemma}\label{lemma0}
The capacity region $\C (V,T,\bar P)$ is symmetric with respect to the line $r_1=r_2$.
\end{lemma}
\begin{proof}
Suppose that a rate-pair $(\check{r}_1,\check{r}_2)\in \C (V,T,\bar P)$ is achieved by $\check{\Q}=\{ \check{x}(t),  t\in \T\}$ and $\check{\pow}=  \{\check{p}_{k}(t), t\in \T, k\in\{1,2\}\}$. Then we can construct another solution $\hat{\Q}=\{\hat x(t),t\in \T\}$ and $\hat{\pow}=\{\hat p_k(t), t\in \T, k\in\{1,2\}\}
$ with $\hat x(t) = -\check{x}(t)$, $\hat p_1(t)=\check{p}_2(t)$, and $\hat p_2(t)=\check{p}_1(t)$, $\forall\, t$, which can be easily shown to achieve the symmetric rate-pair $(\check{r}_2,\check{r}_1)$. As the newly constructed solution is also  feasible to (P1), this lemma is thus proved.
\end{proof}

Based on Lemma \ref{lemma0}, it is evident  that the boundary of $\C (V,T,\bar P)$ is also symmetric with respect to the line $r_1=r_2$.
\begin{lemma}\label{lemma1}
For problem  (P1),  the optimal UAV trajectory satisfies $x^*(t) \in [-D/2, D/2], \forall\, t\in \T$, i.e., the UAV should stay above the line segment between the two GUs.
\end{lemma}
\begin{proof}
Supposing that  the optimal UAV trajectory does not lie within the interval $[-D/2, D/2]$, we can always construct a new trajectory with $x^*(t) \in [-{D}/{2}, {D}/{2}], \forall\, t\in \T$, which simultaneously decreases the distances from  the UAV to both GUs,  thus resulting in a strictly componentwise larger rate-pair based on \eqref{eq:channelgain}  and \eqref{MAC:ineq1}-\eqref{MAC:ineq3}. This thus completes the proof.
\end{proof}

\begin{lemma}\label{lem:flydirection}
For problem  (P1), there always exists an optimal UAV trajectory $\{x^*(t)\}$ that is unidirectional, i.e.,
$x^*(t_1) \le x^*(t_2)$  if $t_1<t_2$, $\forall\, t_1, t_2\in \T$.
\end{lemma}
\begin{proof}
Suppose that $\{x^*(t)\}$ is a non-unidirectional optimal UAV trajectory to problem (P1), which implies that the UAV visits some locations more than one times. We denote the total time that the UAV stays at such a location $x_A$ ($ \min\limits_{t\in \T}  \{x^*(t)\}\le x_A\le \max\limits_{t\in \T} \{x^*(t)\}$)  by $\delta_{t,A}>0$. Then, we show that there always exists an alternative unidirectional UAV trajectory that achieves the same objective value of (P1).   
 Specifically, we construct a unidirectional UAV trajectory $\{\hat{x}(t)\}$ with $\hat{x}(0)= \min\limits_{t\in \T}  \{x^*(t)\}$ and $\hat{x}(T)= \max\limits_{t\in \T} \{x^*(t)\}$ as its initial and final locations, respectively, where the UAV stays at location $x_A$ with a duration $\delta_{t,A}$, i.e., $\hat{x}(t)=x_A$, $t\in [t_A, t_A+\delta_{t,A}]$, with $t_A$ being the time instant once the UAV reaches location $x_A$.
  It is easy to show that $\hat{x}(t)$ is feasible to (P1) and always achieves the same objective value as $x^*(t)$. This thus completes the proof.
\end{proof}

With Lemmas \ref{lemma1} and \ref{lem:flydirection}, we only need to consider the unidirectional UAV trajectory between  $[-{D}/{2}, {D}/{2}]$ in the rest of the paper.
%
In addition, as implied by Lemma \ref{lemma0}, we only need to obtain the boundary point $(r^*_1,r^*_2)$ with $r^*_2\ge r^*_1$ at one side of the line $r_1 = r_2$. This corresponds to the optimal solution to problem (P1) in the case with $\alpha_2 \ge \alpha_1$.
When $\alpha_1=0$ or $\alpha_2=0$, it is easy to show that the optimal rate-pair is  $(r^*_1, r^*_2)= (0, \log_2(1+ \frac{\bar P\beta_0}{H^2}))$ or $(r^*_1, r^*_2)=(\log_2(1+ \frac{\bar P\beta_0}{H^2}), 0)$ where $\beta_0 \triangleq \frac{\gamma_0}{\sigma^2}$.

Next, we introduce a simple and yet practical HFH UAV trajectory, which will be shown optimal for (P1) in the sequel.
Specifically, with the HFH trajectory, the UAV successively  hovers at a pair of  initial and final locations, denoted by $x_{\rm I}$ and $x_{\rm F}$, respectively, with $-D/2\leq x_{\rm I}\leq x_{\rm F}\leq D/2$,  each for a certain amount of time, denoted by $t_{\rm I}\geq 0$ and $t_{\rm F}\geq 0$ with $t_{\rm I} +t_{\rm F}\leq T$,   and flies at the maximum speed $V$ between them. Mathematically, the HFH trajectory is generally given by
\begin{align}\label{optimal:trj}
x(t)=\left\{\begin{aligned}
                  &     x_{\rm I},\qquad \qquad\qquad ~ t \in \T_1, \\
                      & x_{\rm I} + (t-t_{\rm I})V,  ~~~~ t \in \T_2,\\
                   &    x_{\rm F}, \qquad \quad \quad ~~~~~~   t \in \T_3,
                    \end{aligned}
             \right.
\end{align}
where 
$\T_1= [0,t_{\rm I} ]$,  $\T_2 =  (t_{\rm I}, T-t_{\rm F})$, $\T_3= [T-t_{\rm F},T]$, and $t_{\rm I}  + \frac{x_{\rm F}-x_{\rm I}}{V} + t_{\rm F}=T$. It follows from \eqref{optimal:trj} that under the HFH trajectory, the UAV hovers at two different locations  at most; while if $x_{\rm I}=x_{\rm F}$,  then the UAV trajectory  reduces to hovering  at a fixed location during the entire flight duration $T$.
In the following two sections, we first investigate the solutions to (P1) for the two special cases with $T\to \infty$ and $V\rightarrow0$, respectively.


\section{Capacity  Characterization with Large Flight Duration}\label{section:infinite}
In this section, we study the special case when the UAV flight duration is asymptotically  large, i.e., $T\rightarrow \infty$, where the corresponding capacity region is denoted by  $\C ( V,\infty,\bar P)$. To this end, we first ignore  the UAV maximum speed constraint \eqref{P1:speed} in (P1) and derive its optimal solution for any $T>0$. Then, we show that the resulting capacity region is equal  to  $\C ( V,\infty,\bar P)$ as $T\rightarrow \infty$.

By dropping  constraint \eqref{P1:speed} under finite $T$, problem (P1) is reduced to
\begin{align}
\text{(P2)}: ~~\max_{\{r,r_1,r_2,\Q, \pow\}} & r \\
\mathrm{s.t.}~~~~& \eqref{eq:9},\eqref{P1:11}, \eqref{P1:12}, \eqref{P1:14}, \label{P2:eq:9}
\end{align}
whose optimal objective value serves as an upper bound of that of problem (P1).
Although problem (P2) is a non-convex optimization problem, we obtain its optimal solution as in the following lemma.
\begin{lemma}\label{lemm:P2}
Under given $(\alpha_1, \alpha_2)$ with $\alpha_1 + \alpha_2 = 1$, the optimal trajectory and power allocation solution to (P2) is given as $x^*(t) = -D/2$, $p^*_1(t) = \bar P$, $p^*_2(t) = 0$, $\forall t \in \hat{\T}_1$, and $x^*(t) = D/2, p^*_2(t) = \bar P, p^*_1(t) = 0, \forall t \in \hat{\T}_2$, where
$\hat{\T}_1 = [0, \alpha_1T)$, and $\hat{\T}_2 = [\alpha_1T,  T]$. Accordingly, the optimal rate-pair is obtained as $r^*_1 = \alpha_1\log_2(1+ \frac{\bar P\beta_0}{H^2})$ and $r^*_2 =  \alpha_2\log_2(1+ \frac{\bar P\beta_0}{H^2})$.
\end{lemma}
\begin{proof}
Please refer to Appendix A. 
\end{proof}

Based on Lemma \ref{lemm:P2} and by changing the values of $\alpha_1$ and $\alpha_2$ for (P2),   the capacity region without considering the UAV maximum speed constraint \eqref{P1:speed}, denoted by $\hat{\C}(\bar P)$, can be easily obtained in the following proposition.


\begin{proposition}\label{prop1}
In the absence of constraint \eqref{P1:speed},  the capacity region $\hat{\C}(\bar P)$ of the UAV-enabled two-user BC is given by
\begin{align}\label{eq:region1}
\hat{\C}(\bar P)=  \bigg \{ (r_1,r_2) : & ~ r_1+r_2 \leq\log_2\left(1+ \frac{\bar P\beta_0}{H^2}\right),  r_1\geq 0,  r_2\geq 0\bigg\},
 \end{align}
 which is an equilateral triangle.
\end{proposition}

From Lemma  \ref{lemm:P2} and Proposition \ref{prop1}, the optimal UAV trajectory for achieving the boundary points of $\hat{\C}(\bar P)$ is to let the UAV  successively hover above each of the two GUs for communication  in a TDMA manner. It is worth pointing out that the UAV maximum speed is finite in practice and thus constraint  \eqref{P1:speed} cannot be ignored in general.  As a result,  the capacity region $\hat{\C}(\bar P)$ in \eqref{eq:region1} generally serves as an ``upper bound'' of the capacity region with finite $V$. However, as shown in the following theorem, $\hat{\C}(\bar P)$ can be asymptotically achieved when $T$ is sufficiently large for any $V>0$.

\begin{theorem}\label{achievingTDMA}
As $T\rightarrow \infty$, we have $\C ( V,\infty,\bar P) = \hat{\C}(\bar P)$, $\forall\, V>0$, where the optimal UAV trajectory follows the HFH structure  in \eqref{optimal:trj} with $x_{\rm I}=-D/2$ and $x_{\rm F}=D/2$, and the TDMA-based transmission is capacity-achieving.
\end{theorem}
\begin{proof}
First,  it is evident that $\C(V,T,\bar P) \subseteq \hat{\C}(\bar P)$ for any $V>0$ and $T>0$. Next, we show that the HFH trajectory in \eqref{optimal:trj} with $x_{\rm I} = -D/2$ and $x_{\rm F} = D/2$ together with TDMA-based transmission achieves  the boundary of $\hat{\C}(P)$ as $T \rightarrow \infty$, as follows.
 For any boundary point $(r^{**}_1, r^{**}_2) \in \C(V,T,\bar P)$ in  \eqref{eq:region1} satisfying $r^{**}_1 = \alpha_1\log_2( 1+ \frac{\bar P\beta_0}{H^2})$, $r^{**}_2 = \alpha_2\log_2( 1+ \frac{\bar P\beta_0}{H^2})$, $\alpha_1 +\alpha_2 =1$,
 we can construct a feasible solution for (P1)  where the UAV flies at the maximum speed between the two GUs and hovers above GUs 1 and 2 for $\alpha_1$ and $\alpha_2$ proportion of the remaining time. In addition, the UAV only transmits information to  GU 1 or 2 when hovering above that  GU via TDMA. Thus, the corresponding achievable rate-pair of the two GUs, denoted by $(r_1^{\star},r_2^{\star})$, are given by
 $r_1^{\star} = \alpha_1(1-\frac{D}{VT})\log_2( 1+ \frac{\bar P\beta_0}{H^2})$ and  $r_2^{\star} = \alpha_2(1-\frac{D}{VT})\log_2( 1+ \frac{\bar P\beta_0}{H^2})$, where $\frac{D}{VT}$ corresponds to the proportion of time for the UAV's maximum-speed flying from $x_{\rm I} = -D/2$ to $x_{\rm F} = D/2$. As $T\rightarrow \infty$, $(r_1^{\star} , r_2^{\star}) \rightarrow (r^{**}_1, r^{**}_2)$ for any $V>0$ and $D>0$ since $\frac{D}{VT}\rightarrow 0$. Based on the facts that  $(r_1^{\star} , r_2^{\star})\preceq (r_1^{*} , r_2^{*})$ and  $(r_1^{*} , r_2^{*})\preceq (r_1^{**} , r_2^{**})$ (due to $\C(V,T,\bar P) \subseteq \hat{\C}(\bar P)$), we have    $ (r_1^{*} , r_2^{*}) \rightarrow (r^{**}_1, r^{**}_2)$  as  $T\rightarrow \infty$. Thus, the unidirectional HFH trajectory with $x_{\rm I}=-D/2$ and  $x_{\rm F}=D/2$ with TDMA  is asymptotically optimal and  $\C (V, \infty,\bar P) = \hat{\C}(\bar P)$.
\end{proof}
%
\begin{figure}[!t]
\centering
\includegraphics[width=0.5\textwidth]{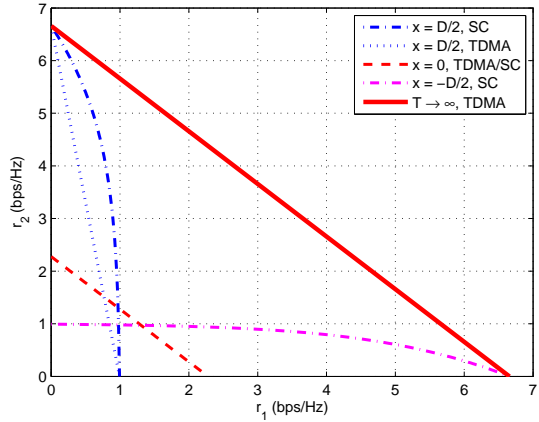}
\caption{The capacity region of a UAV-enabled two-user BC when $T\to \infty$.
} \label{region:infty}\vspace{-0.6cm}
\end{figure}

For the purpose of illustration,  Fig. \ref{region:infty} shows  the capacity region $\C (V, \infty,\bar P)$ with $\sigma_0=-100$ dBm, $\gamma_0= -50$ dB,  $H = 100$ m, $D = 1000$ m, and $\bar P= 10$ dBm. For comparison, we also show the capacity region achieved when the UAV is fixed at a given location $x$, where $x=-D/2$, $0$, or $D/2$. In this case,  the system becomes a conventional two-user AWGN  BC with constant channel gains and its capacity region is denoted by  $\mathcal{C}_{f}(x)$. 
Based on the uplink-downlink duality \cite{jindal2004duality},  we have  $\mathcal{C}_{f}(x)= \underset{p_1+p_2\le \bar P, p_1, p_2\ge 0}{ \bigcup} \C_{\rm{MAC}} (x, p_1,p_2)$, where $\C_{\rm MAC} (x, p_1,p_2)$ denotes the capacity region of the dual two-user MAC specified by the following  inequalities  \cite{tse1998multiaccess} (or equivalently $\C(\Q,\pow)$ with $x(t) = x, p_k(t) = p_k, \forall t\in \T, k\in\{1,2\}$):
 \begin{align}
r_1 \le & \log_2\big(1+ p_1h_1(x)\big), \label{eq:fixed1}\\
r_2  \le & \log_2\big(1+ p_2h_2(x)\big),  \label{eq:fixed2}\\
r_1 + r_2\le & \log_2\big(1+ p_1h_1(x) +p_2h_2(x)\big)  \label{eq:fixed3},
\end{align}
where $h_k(x) = \frac{\beta_0}{(x- x_k)^2+H^2},  k\in\{1,2\}$. It is known  that the capacity region $\mathcal{C}_{f}(x)$ is  convex and its boundary is generally achieved by  SC-based non-orthogonal transmission with interference cancellation at the receiver of the  GU with higher channel gain (or nearer to the UAV in our context) \cite{tse2005fundamentals}.\footnote{For degraded BC, dirty paper coding (DPC) also achieves the same capacity boundary as SC  \cite{tse2005fundamentals}. Without loss of optimality, we consider SC in this paper.}
For example, in Fig. \ref{region:infty}, when $x=D/2$ or $x=-D/2$, the corresponding boundary points of $\mathcal{C}_{f}(x)$ can only be achieved by applying SC while TDMA is strictly suboptimal (except for the two extreme points)\cite{tse2005fundamentals}. By contrast, when $x=0$, the two GUs have the same channel gain and thus both SC and TDMA  are optimal.
Interestingly, based on Theorem \ref{achievingTDMA}, when the UAV mobility can be fully exploited (say, with untethered UAVs)  with $T\to \infty$ (or equivalently $D\ll VT$),  TDMA-based orthogonal transmission along   with the simple HFH UAV trajectory is capacity-achieving whereas SC  is not required.
Nevertheless, it is also worth pointing out that the significant capacity gain by exploiting the high mobility  UAV over the static UAV comes at the cost of transmission delay at one of the two GUs (e.g.,  GU 2 needs to wait for about $T/2$ to be scheduled for transmission, which can be substantial when $T$ becomes large). Therefore, there is a fundamental throughput-delay trade-off in wireless communications enabled by high-mobility UAVs \cite{lyu2016cyclical, JR:wu2017_ofdm}.

\begin{remark}
Based on Proposition \ref{prop1}, we next provide a property of capacity region for finite flight duration, which helps reveal the fundamental reason why the UAV mobility can potentially enlarge the  capacity region: $\C (V,T,\bar P)$ is non-convex  for $0< T< \infty$.
The non-convexity of the capacity region essentially suggests that it can be enlarged  if the convex combinations of the rate-pairs can be achieved (generally achieved at different locations). This is the reason why the UAV mobility/movement can be helpful, leading to a time-sharing of multiple locations.
\end{remark}

\vspace{-0.25cm}
\section{Capacity Characterization with Limited UAV Mobility}
In this section, we study the other special case with limited UAV mobility, i.e.,  $V\rightarrow 0$ (or equivalently $VT\ll D$). In this case,  the UAV's horizontal movement has negligible impact on the UAV-GU channels with $H\gg VT$ (since $H$ is comparable with $D$). As a result, the UAV should hover at a fixed location during the entire $T$ once it is deployed (e.g., a tethered UAV), i.e.,  $x(t) = x, \forall\, t\in\mathcal T$, which becomes a special case of  the proposed HFH trajectory in \eqref{optimal:trj} with $x_{\rm I}=x_{\rm F}$. In this case, solving (P1) is equivalent to finding the optimal hovering location of the UAV, $x$, given the rate-profile parameters $(\alpha_1, \alpha_2)$, as well as the corresponding transmission power,
$p_1(t_1) =  p_1$ and $p_2(t_1) =  p_2$, $\forall\, t\in\mathcal T$, and rates $r_1$ and $r_2$ for GUs 1 and 2, respectively. As such,   $\C(\Q, \pow) = \C_{\rm{MAC}} (x, p_1,p_2)$ holds and problem (P1) is reformulated as
\begin{align}
\text{(P3)}: ~~\max_{\{r,r_1,r_2, x, p_1,p_2\}} & r \\
\mathrm{s.t.}~~~~& r_k \ge \alpha_k r, \forall\, k\in\{1,2\},\label{P3:eq01}\\
&(r_1,r_2) \in \C_{\rm MAC}(x, p_1,p_2),\label{P3:eq02}\\
&p_1 + p_2 \le \bar P, \label{P3:eq03}\\
&p_k \geq 0,   k\in\{1,2\}. \label{P3:eq05}
\end{align}
\begin{proposition}\label{thm2}
For problem (P3) with $\alpha_2 \ge \alpha_1$, there always exists an optimal  UAV hovering location  $x^*$, such that $0 \le x^* \le D/2$.
\end{proposition}
\begin{proof}
Please refer to Appendix B. 
\end{proof}

Proposition  \ref{thm2}  suggests that when $\alpha_2 \ge \alpha_1$, the UAV should be placed closer to  GU 2 such that  it has a larger channel gain than  GU 1. As a result,   GU 2 needs to decode  GU 1's signal first and then decodes its own signal after canceling the interference from  GU 1's signal. Therefore, we have $r_1 = \log_2(1+\frac{p_1h_1(x) }{p_2h_1(x)+1})$ and $r_2 =  \log_2\left(1+p_2h_2(x)\right)$, where $h_1(x)= \frac{\beta_0}{ (x+\frac{D}{2})^2+H^2 }$ and $h_2(x)= \frac{\beta_0}{ (x-\frac{D}{2})^2+H^2 }$.
As the inequalities in \eqref{P3:eq01} must be tight at the optimality, (P3) can be transformed into the following problem,
\begin{align}
\text{(P4)}: ~~ \max_{\{x,p_2\}} ~&\frac{1}{\alpha_2}\log_2\left(1+p_2h_2(x)\right)  \\
\mathrm{s.t.}~& \alpha_1 \log_2\left(1+p_2h_2(x)\right) =\alpha_2 \log_2\left(1+\frac{(\bar P-p_2)h_1(x) }{p_2h_1(x)+1}\right),\label{eq:28}\\
&0 \le p_2\le \bar P,~ x\in [0,D/2 ].
\end{align}
Note that for $\alpha_k>0$, $k\in \{1,2\}$,  the LHS  (RHS) of \eqref{eq:28} increases (decreases) with $p_2$.  It thus follows  that under any given $x$, the optimal solution of $p_2$ is unique and can be directly obtained by solving the equality in \eqref{eq:28} with a bisection search, and thus the objective value of (P4) can be obtained accordingly.
Therefore, to solve problem (P4), we only need to apply  the one-dimensional search over $x\in [0,D/2]$, together with a bisection search for $p_2$ under each given $x$.
\begin{figure}[!t]
\centering
\includegraphics[width=0.5\textwidth]{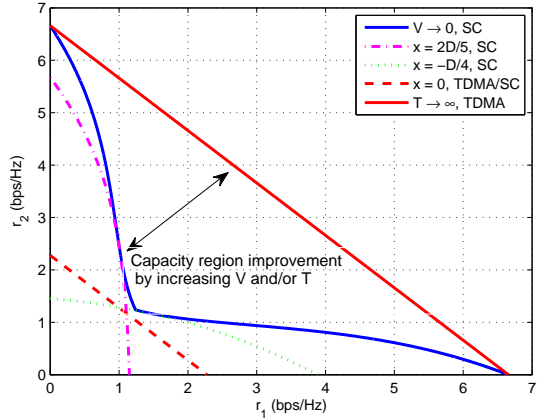}
\caption{The capacity region of a UAV-enabled two-user BC when $V\rightarrow0$.
} \label{region:zero}\vspace{-0.5cm}
\end{figure}

Fig. \ref{region:zero} shows the capacity region $\C (0,T,\bar P)$ of the UAV-enabled two-user BC as $V\rightarrow 0$. The parameters are same as those for Fig. \ref{region:infty}. It is observed that $\C (0,T,\bar P)$ is a non-convex set that is larger than the two-user AWGN BC capacity region $\C_f(x)$ at any fixed location $x$, thanks to the location optimization for the UAV based on the  GU rate requirements (or rate-profile vector). In addition, it is interesting to observe  that at some locations,  e.g., $x = 2D/5$ and $x = -D/4$, the fixed-location capacity region $\C_f(x)$ touches the boundary of $\C (0,T,\bar P)$,  while at other locations, e.g., $x = 0$, $\C_f(x)$ lies strictly inside $\C (0,T,\bar P)$.
The latter observation  suggests that some locations are inferior  to the others in the sense that they achieve componentwise smaller rate-pairs.
This  further  implies that when $V>0$,  the UAV should hover at such superior locations rather than the inferior locations in order to maximize the  GU average rates. This is illustrated by observing $\C_f(0) \subset  \C_f(2D/5) \bigcup \C_f(-D/4)$  in Fig. \ref{region:zero}.   However, since the UAV's speed is finite in practice, it may need to fly over some inferior locations (e.g., $x=0$), in order to travel between and hover over different superior locations that are far apart (e.g., $x=2D/5$ and $x=-D/4$) in a time-sharing manner. This intuitively explains  why the UAV should fly at the maximum speed in the optimal HFH trajectory for the general case with finite UAV maximum  speed $V$ and flight duration $T$, as will be rigorously proved in the next section.
Finally, by comparing $\C (0,T,\bar P)$ with $\C (V, \infty,\bar P)$, it is observed that significant capacity improvement can be achieved  by increasing the UAV maximum speed and/or flight duration.

\section{Capacity Characterization for Finite UAV Speed and Flight Duration}\label{section:capacity:region}
In this section, we characterize the capacity region by solving problem (P1) for the general case with finite UAV maximum speed $V$ and flight duration $T$. 
\vspace{-0.3cm}
\subsection{Capacity Region Characterization}
First, we reveal an important property of the optimal UAV trajectory solution to problem (P1) with any given $V>0$ and $T>0$, based on which, we show that  the HFH UAV trajectory is capacity-achieving. According to Lemmas \ref{lemma1} and \ref{lem:flydirection}, we consider a unidirectional UAV trajectory without loss of generality, in which the initial and final locations are denoted by $x^*(0)= x_{\rm I}$ and $x^*(T) = x_{\rm F}$, respectively, with $-D/2 \leq x_{\rm I}\leq x_{\rm F}\leq D/2$. 

Intuitively, the fixed-location capacity regions of  $x_{\rm I}$ and $x_{\rm F}$, i.e., $\C_f(x_{\rm I})$ and $\C_f(x_{\rm F})$,  should have rate superiority over those of the other locations on the line between them, which is affirmed by the following proposition.

\begin{proposition}\label{trj_lemma}
At the optimal UAV trajectory solution $\{x^*(t)\}$ to problem (P1), it must hold that
\begin{align}\label{eq:V1}
\C_f(x_A) \subseteq {\rm{Conv}}\big(  \C_f(x_{\rm I}) \bigcup \C_f(x_{\rm F})\big),
\end{align}
 for any location $x_A$ between the initial and final locations $x_{\rm I}$ and $x_{\rm F}$, i.e., $x_{\rm I}\le x_A\le x_{\rm F}$.
\end{proposition}
\begin{proof}
Please refer to Appendix C. 
\end{proof}
\begin{figure}[!t]
\centering
\includegraphics[width=0.35\textwidth]{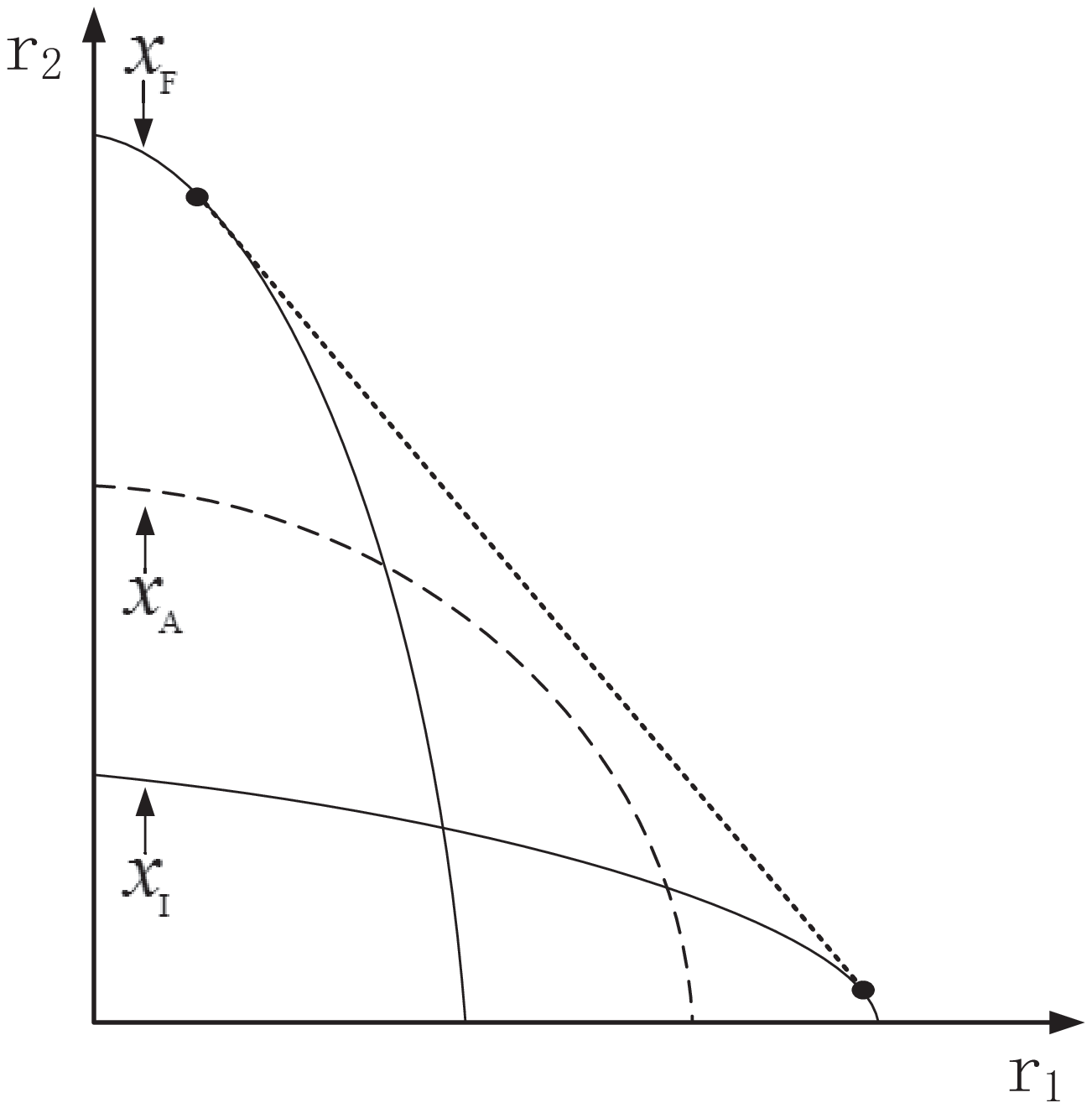}
\caption{Illustration of  Proposition \ref{trj_lemma} where $\C_f(x_A) \subseteq {\rm{Conv}}\big(\C_f(x_{\rm I})\bigcup \C_f(x_{\rm F}) \big)$ for any location $x_A$ between $x_{\rm I}$ and $x_{\rm F}$.
} \label{region:infty0}\vspace{-0.7cm}
\end{figure}

Proposition \ref{trj_lemma} essentially implies  there always exists  a rate-pair  in the boundary of the convex hull, $\C_f(x_A) \subseteq {\rm{Conv}}\big(\C_f(x_{\rm I})\bigcup \C_f(x_{\rm F}) \big)$,  which is componentwise no smaller than any given rate-pair in the fixed-location capacity region at location  $x_A$ (i.e., $\C_f(x_A)$), as illustrated by Fig. \ref{region:infty0}.
Based on Proposition \ref{trj_lemma},  the optimal UAV trajectory to problem (P1) is obtained as follows.


\begin{theorem}\label{prop4}
For problem (P1) with given $T>0$ and $V>0$, the HFH trajectory in \eqref{optimal:trj} is optimal.
\end{theorem}
\begin{proof}
Please refer to Appendix D. 
\end{proof}


Based on Theorem \ref{prop4}, the optimal UAV trajectory $\{x^*(t)\}$ is determined only by the initial and final hovering locations $x_{\rm I}$ and $x_{\rm F}$ as well as  the hovering time $t_{\rm I}$ at location $x_{\rm I}$. Accordingly, the optimal hovering time  $t_{\rm F}$ can be obtained as $t_{\rm F}=T-t_{\rm I}  - \frac{x_{\rm F}-x_{\rm I}}{V}$. Therefore, we can solve problem (P1) by first optimizing the power allocation under any given UAV trajectory, and then searching over the three variables $x_{\rm I}$, $x_{\rm F}$, and $t_{\rm I}$  to obtain the optimal UAV trajectory for any given $(\alpha_1, \alpha_2)$.
Specifically, based on a fixed HFH UAV trajectory $\{x(t)\}$, (P1) is reduced to the following problem,
\begin{align}
\text{(P5)}: ~~\max_{\{r,r_1,r_2, \pow \}} & r \\
\mathrm{s.t.}~&  r_k \ge \alpha_k r, \forall\, k\in\{1,2\},\label{eq:99}\\
&(r_1,r_2) \in \mathcal{C}(\Q, \pow ),\\
&p_{1}(t)+ p_{2}(t)\leq \bar P, \forall\,n,\\
&p_k(t) \geq 0, \forall\, t\in \T,  k\in\{1,2\}.
\end{align}
Note that (P5) is a convex optimization problem and thus can be solved efficiently by applying the well-established polymatroid structure and the Lagrange duality method \cite{tse1998multiaccess,Boyd}. Therefore, the optimal rate-pair $(r_1^*, r_2^*)$ corresponding to rate-profile $(\alpha_1,\alpha_2)$ can be found by applying a three-dimensional search on $x_{\rm I}$, $x_{\rm F}$, and $t_{\rm I}$, and selecting their  values to maximize $r$ in (P5). The details are omitted here due to the space limitation. Notice that in this case, SC-based transmission is generally required.
\vspace{-0.5cm}
%
\begin{figure*}[!t]
\begin{minipage}[t]{0.50\linewidth}
\centering
\includegraphics[width=3.2in, height=2.5in]{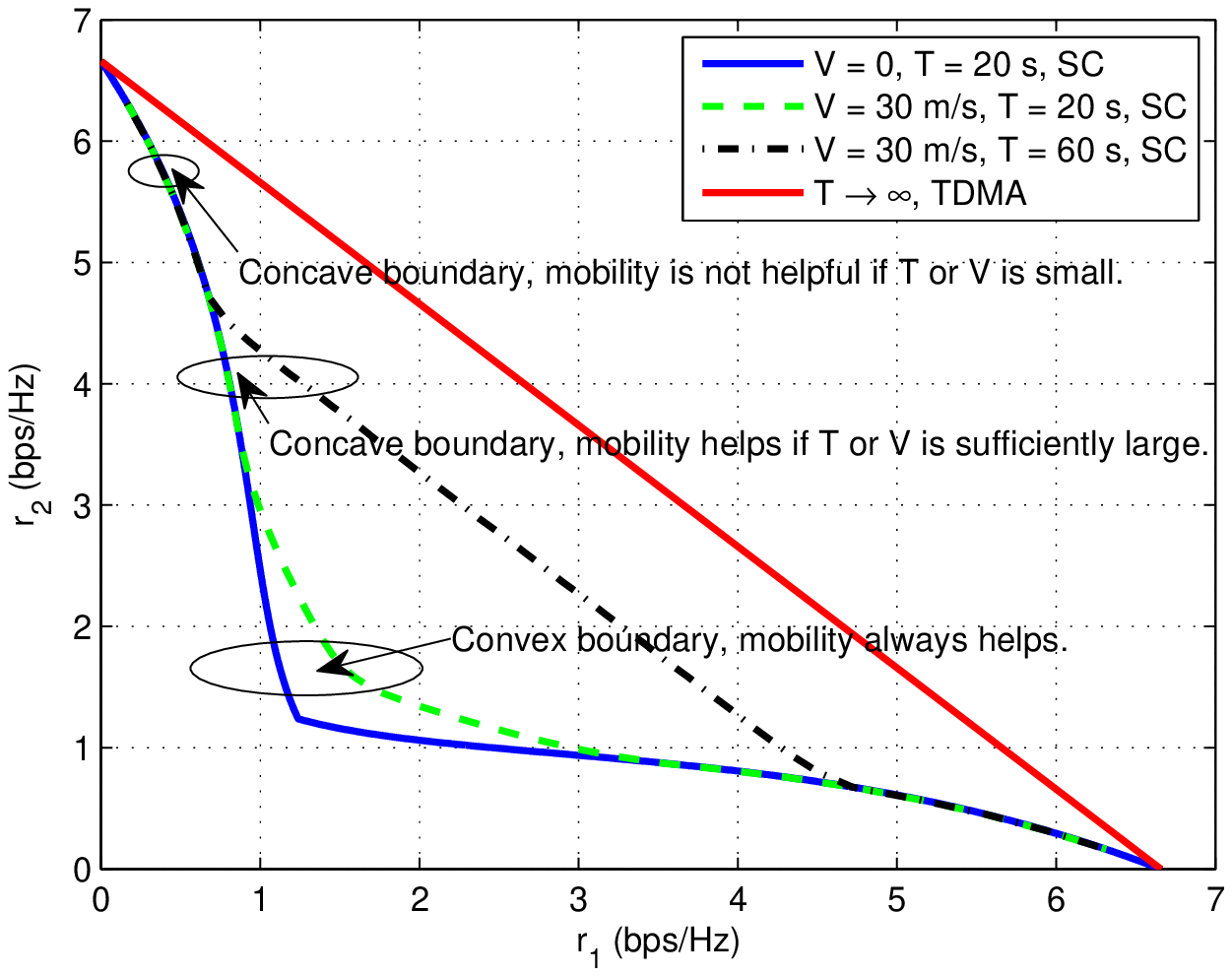}  
\caption{The capacity region of a UAV-enabled two-user BC with $\bar P = 10$ dBm.}
\label{region:CDMA:mobile}
\end{minipage}%
~~
\begin{minipage}[t]{0.50\linewidth}
\centering
\includegraphics[width=3.2in, height=2.5in]{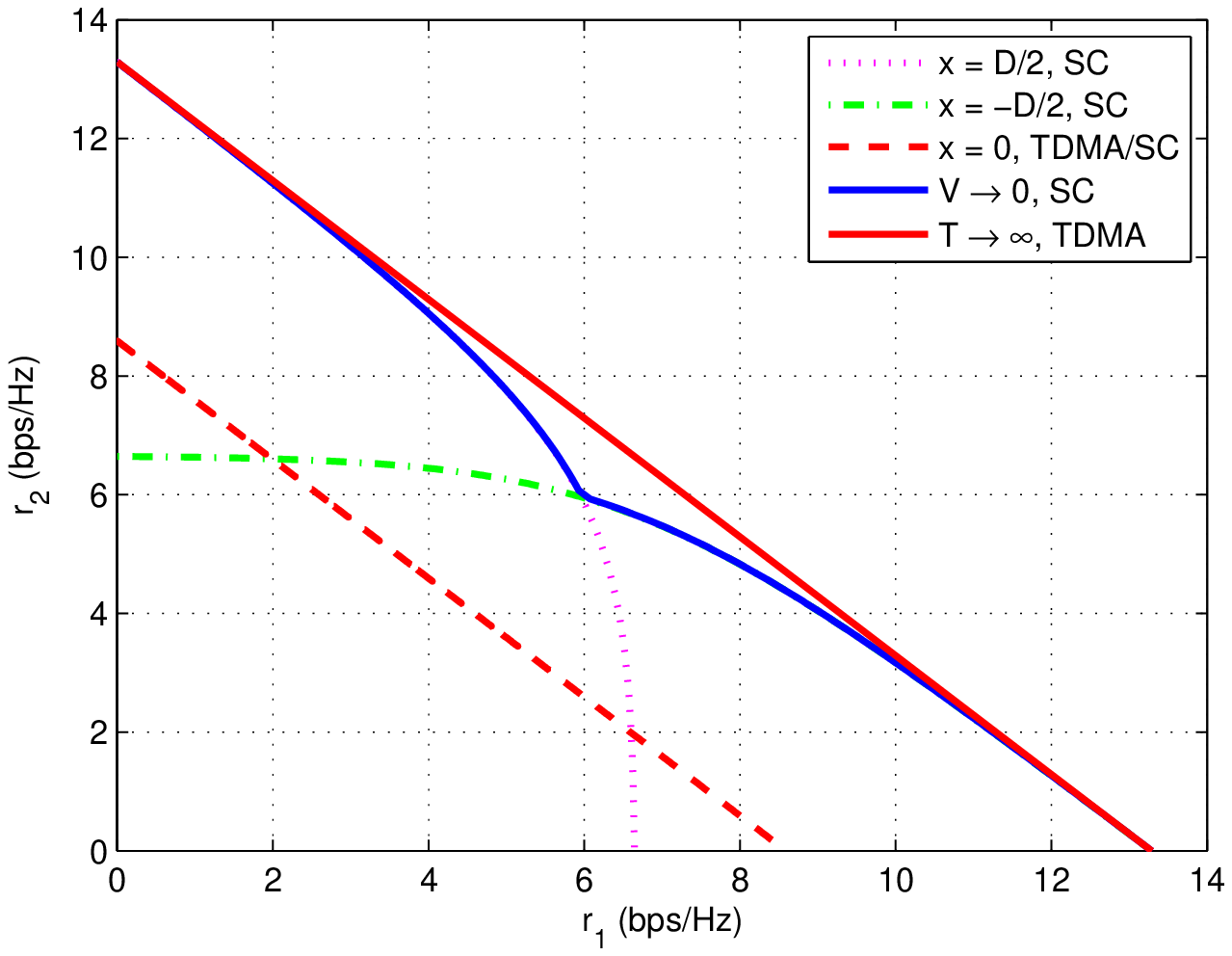}
\caption{The capacity region of a UAV-enabled two-user BC with $\bar P = 30$ dBm. }
\label{region:highSNR}
\end{minipage}\vspace{-0.5cm}
\end{figure*}

\subsection{Numerical Results}
 In Fig. \ref{region:CDMA:mobile},  the capacity region $\C (V,T,\bar P)$ for finite UAV maximum speed and flight duration is shown under different setups.  
 The parameters are same as those for Figs. \ref{region:infty} and \ref{region:zero}. It is interesting to observe that although $\C (V,T,\bar P)$ is generally a non-convex set, its  boundary has a general  \emph{concave-convex-concave} shape when $V$ and $T$ are finite. 
  This observation helps explain whether the UAV movement is able to enlarge the capacity region or not. Specifically, when $V=0$, the boundary is convex for  $r_2 \in [1, 3]$ bps/Hz, which implies that the convex combination of any two boundary points (rate-pairs) in this regime always achieves a componentwise larger rate-pair than any boundary points between them. As such, increasing the UAV maximum speed and/or flight duration enables the UAV to fly closer to each  GU and thus achieves higher rate-pairs, leading to an enlarged capacity region. By contrast, the boundary is concave for  $r_2 \in [3.5,  6]$ bps/Hz,  which means that the convex combination of any two boundary points within this regime will achieve a componentwise smaller rate-pair than any boundary points between them. This suggests that if $V$ and $T$ are small such that the UAV can only fly locally among these locations,
   it is not desirable  for the UAV to move in terms of achieving a componentwise larger rate-pair. This is in fact the reason why in Fig. \ref{region:CDMA:mobile}, when $r_2 \in [3.5,  6]$ bps/Hz,  the boundary  for $V=30$ m/s and $T=20$ s remains the same as that for $V=0$. However, when the duration $T$ is further increased from $20$ s to $60$ s, it is observed that the boundary for $r_2 \in [3.5, 4.5]$ bps/Hz shifts towards the upper-right direction, which means that the UAV movement becomes helpful. This is because with sufficiently large $T$, the UAV is able to fly over its nearby locations to reach some superior locations and hence can achieve a componentwise larger rate-pair. In fact, with any given $V>0$, as long as $T$ is sufficiently large, the UAV movement is always beneficial to enlarge the capacity region.



\vspace{-0.25cm}
\subsection{High SNR Case}
Lastly, we consider the  asymptotically high SNR case with $\bar{P}\rightarrow \infty$ such that  $\frac{\bar P \beta_0}{ D^2 + H^2 }\gg1$ can be assumed, to provide more insights. This assumption means that if the UAV is placed above one (near)  GU, the SNR of the other (far)  GU (with maximum UAV-GU distance $\sqrt{D^2+H^2}$) is still sufficiently large when $\bar P$ is used.
 \begin{theorem}\label{prop5}
Under the assumption of $\frac{\bar P \beta_0}{ D^2 + H^2 }\gg1$,    the optimal HFH UAV trajectory  to (P1) is simplified to  $x^*(t) = D/2,\forall\, t\in \T$ if $\alpha_2 \ge \alpha_1$; and   $x^*(t) = -D/2, \forall\, t\in \T$  if $\alpha_2 < \alpha_1$.
Accordingly, the capacity region  is given by
  \begin{align}\label{eq:highSNR}
\mathcal{C}_{\rm h-SNR}(V, T, \bar P) =  \bigg \{ (r_1,r_2) : & ~ r_1+r_2\leq \log_2\left(\frac{\bar P \beta_0}{ H^2 }\right), r_1\geq 0,  r_2\geq 0\bigg\}.
 \end{align}
 \end{theorem}
\begin{proof}
Please refer to Appendix E. 
\end{proof}

 Similar to Proposition   \ref{prop1}, Theorem \ref{prop5}  shows the superiority of the UAV hovering locations right above the GUs. However, unlike $\mathcal{C}(V, T, \bar P)$,  the capacity region $\mathcal{C}_{\rm h-SNR}(V, T, \bar P)$ is independent of both UAV flight duration $T$ and maximum speed $V$, which suggests that the UAV movement is less effective to enlarge the capacity region as the SNR becomes large.

 In Fig. \ref{region:highSNR}, we plot the capacity region $\C (V,T,\bar P)$ for the same setup of Fig. \ref{region:CDMA:mobile} except that $\bar{P}$ is increased from $10$ dBm to  $30$ dBm. In this case, we have $\frac{\bar P \beta_0}{ D^2 + H^2 }\approx 100 \gg1$, i.e., the high SNR assumption for Theorem \ref{prop5} approximately holds. It is observed that when $V=0$, the UAV always hovers above the  GU that requires larger rate, e.g., $x=D/2$, for all rate-pairs satisfying $r^*_2\ge r^*_1$, and SC is needed.
Furthermore,  it is observed that hovering at the middle location, i.e., $x=0$, where the two GUs have equal channel gains,
 suffers a significant capacity loss in the high SNR regime even for maximizing the equal rate with $r^*_1=r^*_2$.
  Finally, it is observed that the capacity region improvement is very limited by increasing the UAV maximum speed and/or flight duration in this case since the gap between $\mathcal{C}(0, T, \bar P)$ and $\mathcal{C}(V, \infty, \bar P)$ is already very small, i.e., the gain achieved by exploiting the UAV movement is not appealing.

 \section{Achievable Rate Region with TDMA }\label{Section:TDMA}
In this section, we consider the UAV-enabled two-user BC with TDMA-based communication.
 \subsection{Achievable Rate Region Characterization }
 For TDMA, the UAV can communicate with at most one  GU at any time instant. Denote by $\pi_k(t)\in \{0,1\}, k\in\{1,2\}$ the binary variable which indicates that   GU $k$ is scheduled for communication at time instant $t$ if $\pi_k(t)=1$; otherwise, $\pi_k(t)=0$. Accordingly, the achievable average rate region with  TDMA is given by
 \begin{align}
\bigg\{(r_1, r_2):~& r_1 \le  \frac{1}{T}\int_{0}^T \pi_1(t) \log_2(1+ \bar Ph_1(x(t))) {\rm{d}}t,\label{TDMA:eq1} \\
 & r_2  \le \frac{1}{T}\int_{0}^T\pi_2(t) \log_2(1+ \bar Ph_2(x(t))){\rm{d}}t,\label{TDMA:eq2} \\
&\pi_1(t) +\pi_2(t)\leq 1, \forall\,t \in \T,\label{TDMA:eq3}\\
&\pi_k(t) \in \{0,1\},  \forall\, t\in \T,  k\in\{1,2\}\bigg\}. \label{TDMA:eq4}
 \end{align}
We denote the achievable rate region characterized by \eqref{TDMA:eq1} and \eqref{TDMA:eq2} subject to \eqref{TDMA:eq3} and \eqref{TDMA:eq4} as $\C_{\rm TD} (V,T,\bar P)$. Let  $ \mathbf{\Pi} \triangleq\{ \pi_k(t), t\in \T, k\in\{1,2\} \}$. Similarly as for problem (P1), we can apply the rate-profile approach to characterize $\C_{\rm TD} (V,T,\bar P)$ with rate-profile parameters $(\alpha_1,\alpha_2)$ and the optimization problem is formulated as
{\begin{align}
\text{(P6)}: ~~\max_{\{r,\Q, \mathbf{\Pi}\}} &~ r \\
\mathrm{s.t.}~~~~& \frac{1}{T}\int_{0}^T \pi_1(t) \log_2(1+ \bar Ph_1(x(t))) {\rm{d}}t \ge \alpha_1 r, \label{eq:900}\\
&\frac{1}{T}\int_{0}^T\pi_2(t) \log_2(1+ \bar Ph_2(x(t))){\rm{d}}t \ge \alpha_2 r,  \label{eq:1000}\\
&|\dot{x}(t)| \le V, \forall\, t \in \T, \label{P4:eqspeed}\\
& \eqref{TDMA:eq3}, \eqref{TDMA:eq4}.
\end{align}}Note that problem (P6) is a non-convex optimization problem since it involves binary variables in $ \mathbf{\Pi}$ and  the LHSs of \eqref{eq:900} and \eqref{eq:1000} are not  concave with respect to  $\Q$ even for given $ \mathbf{\Pi}$. Nevertheless, we show in the following proposition that the optimal UAV trajectory still follows the HFH structure as for (P1), except that the UAV only hovers at $x=-D/2$ and/or $x=D/2$.

\begin{proposition}\label{thm:TDMA}
The optimal UAV trajectory to problem  (P6) satisfies the HFH trajectory in \eqref{optimal:trj}. Furthermore, the following properties hold:  1) if $-D/2<x_{\rm I} \le x_{\rm F} < D/2$, then $t_{\rm I}=t_{\rm F}=0$ and $VT=x_{\rm F}-x_{\rm I}$; 2) if $-D/2= x_{\rm I}\le x_{\rm F} <D/2$, then $t_{\rm I}=T- \frac{x_{\rm F}-x_{\rm I}}{V}$ and $t_{\rm F}=0$;  3) if $-D/2<x_{\rm I}\le x_{\rm F} = D/2$, then $t_{\rm I}=0$ and $t_{\rm F}=T- \frac{x_{\rm F}-x_{\rm I}}{V}$; and 4) if $-D/2= x_{\rm I}< x_{\rm F} = D/2$, then $t_{\rm I}+t_{\rm F} +\frac{D}{V}=T$.
\end{proposition}
\begin{proof}
First, similar to the proofs of  Lemmas \ref{lemma1} and \ref{lem:flydirection} and Theorem \ref{prop4},  it can be shown that  the optimal UAV trajectory to problem (P6) satisfies the HFH trajectory in \eqref{optimal:trj} with $x_{\rm I}$ and $x_{\rm F}\in [-D/2, D/2]$.
Therefore, we only need to prove the properties for the above four cases, by contradiction. For brevity, we only consider case 1) at below by showing that the UAV will neither hover at $x_{\rm I}$ nor $x_{\rm F}$, while the other three cases can be verified similarly.

For case 1), suppose that at the optimal UAV trajectory solution to  (P6),  the UAV flies from $x_{\rm I}$ to $x_{\rm F}$ and hovers at location $x_B=x_{\rm I}$ or  $x_B=x_{\rm F}$ for $t_B>0$ with durations $\tau_1>0$ and  $\tau_2>0$ assigned to GUs 1 and 2 for communication, respectively, where $\tau_1+\tau_2= t_B$.  Then, we can construct a new trajectory where the difference is that the UAV flies from $x_{\rm I}'$ to  $x_{\rm F}'$ at the maximum speed with  $x_{\rm I}'=x_{\rm I}-V\Delta\tau_1$ and  $x_{\rm F}'= x_{\rm F}+ V\Delta\tau_2$, where $\Delta\tau_1\in [0, \tau_1]$ and $\Delta\tau_2\in [0, \tau_2]$ are chosen such that $x_{\rm I} > x_{\rm I}'\ge -D/2$ and/or $x_{\rm F}< x_{\rm F}'\le D/2$. Since $x\in [x_{\rm I}', x_{\rm I}]<x_B$ and/or $x\in [x_{\rm F}, x_{\rm F}']>x_B$, it can be shown from \eqref{eq:channelgain} that this newly constructed trajectory can achieve a componentwise larger rate-pair than the assumed one. The proof is thus completed.
\end{proof}

Proposition  \ref{thm:TDMA}  suggests that although the UAV shall  fly at the maximum speed between the initial and final locations as in the capacity-achieving UAV case with SC-based transmission, the hovering locations can only be $x=-D/2$ and $x=D/2$ in the case of TDMA-based transmission. This is quite different from the capacity-achieving UAV trajectory in \eqref{optimal:trj} where the UAV may  hover at  $x\in(-D/2, D/2)$. This is because to achieve the capacity boundary,  SC is generally required and hence there may exist some hovering locations $x\in(-D/2, D/2)$ that can strike a good balance between the channel gains of the two GUs since they will be affected simultaneously (with one decreasing and the other increasing) if the UAV moves. However, since only one  GU will be scheduled at any time for TDMA, the UAV's movement between $(-D/2, D/2)$ can always help increase the channel gain of one  GU (scheduled for transmission) while without degrading that of the other (not scheduled for transmission).

Under the optimal given unidirectional UAV trajectory given in Proposition   \ref{thm:TDMA}, it can be shown that the UAV will first schedule  GU 1 and then schedule GU 2 for transmission, i.e., $\pi_1(t)=1$ and $\pi_2(t)=0$ for $t\in [0, t_1]$, and $\pi_1(t)=0$ and $\pi_2(t)=1$ for $t\in (t_1, T]$ where $t_1\in \T$. 
Based on $\{x(t)\}$ and $t_1$, problem (P4) is reduced to the following problem,
\begin{align}\label{probm:notrj}
\text{(P5)}: ~~\max_{\{r,t_1\}} &~ r \\
\mathrm{s.t.}~~~~& \frac{1}{T}\int_{0}^{t_1}\log_2(1+ \bar Ph_1(x(t))) {\rm{d}}t \ge \alpha_1 r, \label{eq:90}\\
&\frac{1}{T}\int_{t_1}^T \log_2(1+ \bar Ph_2(x(t))){\rm{d}}t \ge \alpha_2 r,  \label{eq:100}\\
&t_1 \in \T.
\end{align}
Note that  the LHS of  \eqref{eq:90} and \eqref{eq:100} increases and decreases with $t_1$, respectively. In addition, constraints \eqref{eq:90} and \eqref{eq:100} need to be met with equalities at the optimal solution. Thus, the optimal $t_1$ is unique and can be obtained efficiently by applying a bisection search over the interval $[0, T]$. As such, the optimal rate-pair $(r_1^*, r_2^*)$ corresponding to rate-profile $(\alpha_1,\alpha_2)$ can be found by applying a two-dimensional search for $x_{\rm I}$ and $t_1$, and  selecting their values to maximize $r$ in  (P5).
\vspace{-0.5cm}
\begin{figure*}[!t]
\begin{minipage}[t]{0.50\linewidth}
\centering
\includegraphics[width=3.2in, height=2.5in]{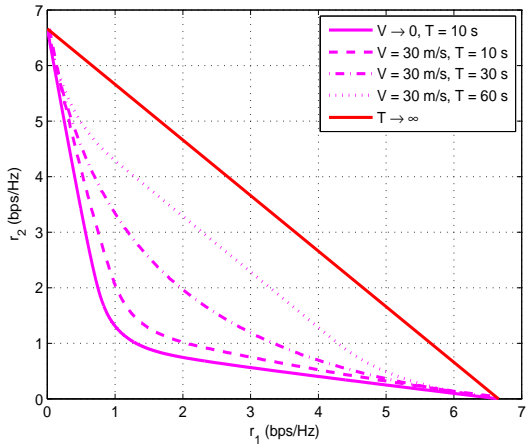}  
\caption{Achievable rate region of a UAV-enabled two-user BC with TDMA,  and $\bar P = 10$ dBm.}
\label{region:TDMA:mobile}
\end{minipage}%
~~
\begin{minipage}[t]{0.50\linewidth}
\centering
\includegraphics[width=3.2in, height=2.5in]{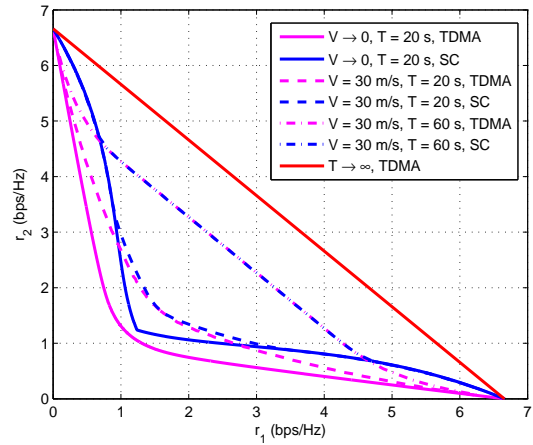}
\caption{Comparison of SC  and TDMA with $\bar P = 10$ dBm. }
\label{region:CD:TD:compare}
\end{minipage}\vspace{-0.5cm}
\end{figure*}
\subsection{Numerical Results}
 We consider the same parameters as those  for Figs. \ref{region:infty}, \ref{region:zero}, and \ref{region:CDMA:mobile}.  In Fig.  \ref{region:TDMA:mobile},  the boundary of  $\C_{\rm TD} (V,T,\bar P)$  is characterized in different setups. It is observed that $\C_{\rm TD} (V,T,\bar P)$ is in general a non-convex set. Furthermore, unlike the boundary of $\C (V,T,\bar P)$ that generally follows a convex-concave-convex shape as shown in Fig. \ref{region:CDMA:mobile}, the boundary  of $\C_{\rm TD} (V,T,\bar P)$ is observed to be always convex. Such a convex  boundary essentially suggests that increasing the UAV speed and/or flight duration is always beneficial to enlarge the achievable rate region of the UAV-enabled BC with TDMA.  This is expected since  with a larger $V$ and/or $T$, the UAV can always fly closer to the  GU that is being scheduled for communication and thus have a better UAV-ground channel in the case of TDMA.

Fig.  \ref{region:CD:TD:compare} compares the capacity region, $\C (V,T,\bar P)$ with SC-based transmission, and the achievable rate region based on TDMA, $\C_{\rm TD} (V,T,\bar P)$. First, it is observed that $\C (V,T,\bar P)$ is generally larger than $\C_{\rm TD} (V,T,\bar P)$ for same values of $V$, $T$, and $\bar P$.
  In particular, for $V=0$ and $r_1= 1$ bps/Hz,  the achievable rate of  GU 2 in  $\C (V,T,\bar P)$ is improved about $80\%$ compared to that in $\C_{\rm TD} (V,T,\bar P)$. Such a rate gain comes not only from using SC  but also from the corresponding UAV location optimization.
  Second, as  $V$ and/or  $T$ increases, it is observed that  the boundary  of  $\C (V,T,\bar P)$ touches more points with that of $\C_{\rm TD} (V,T,\bar P)$ and the gap between them shrinks, suggesting that TDMA becomes more close to be optimal, and they become identical as $T\rightarrow \infty$ (see Theorem \ref{achievingTDMA}).
   For example, when $V=0$, the boundaries of  $\C (V,T,\bar P)$ and $\C_{\rm TD} (V,T,\bar P)$ only intersect at two vertices where the UAV schedules and hovers above only one  GU.  When $V=30$ m/s and $T=60$ s, the boundary of $\C (V,T,\bar P)$ overlaps with that of $\C (V,T,\bar P)$ for $r_1\in[1, 4]$ bps/Hz where the achievable rates of the two GUs are relatively comparable.  This suggests that in this regime, the optimal trajectory for TDMA is also capacity-achieving. This is because when both GUs' rates are non-negligible, it is worth for the UAV flying closer to and even hovering above the two GUs to communicate. However, for a rate-pair in which one of the  GU's rate is very small  but not zero, the UAV only needs to perform proper power allocation to attain this rate-pair rather than spending time in moving towards the  GU with the smaller rate even if $VT>D$. In this case, the optimal UAV trajectory for TDMA in which the UAV greedily flies towards the  GU that is being scheduled, becomes strictly suboptimal compared to that for SC.

\section{Concluding Remarks}
This paper characterizes the capacity region  of a new  two-user BC with a UAV-mounted aerial BS by investigating a joint UAV trajectory and communication design. We show that the optimal rate trade-off of the two GUs or the Pareto boundary of the capacity region is generally achieved by a simple HFH UAV trajectory, under different UAV flight duration and maximum speed, as well as with SC or TDMA-based transmission.
 It is shown that the capacity region can be significantly enlarged by exploiting different forms of mobility of the UAV, via either placement optimization for low-mobility UAVs or HFH trajectory optimization for high-mobility UAVs. In addition, it is shown that TDMA-based design achieves close-to-optimal  performance of SC-based design for sufficiently large UAV speed and/or flight duration and is capacity-achieving when the flight duration goes to infinity.
We hope that the results in this paper for the simplified two-user BC would provide useful insights and guidelines  for designing more general/complex UAV-enabled multiuser communication systems in future, especially from an information-theoretic perspective. In the following, we point out some promising directions to motivate future work.
\begin{itemize}
  \item In this paper, to simplify the analysis, it is assumed that the noise power is equal for the two GUs. For the general case with unequal noise power, the capacity region is not symmetric and thus needs further investigation. In addition, a simplified LoS link is assumed for UAV-GU channels in this paper, while more practical channels models such as Rician fading can be considered in future work \cite{zeng20173d}.

\item Capacity characterization of the general multiuser BC with more than two GUs  is also worth pursuing.  Whether our proposed HFH trajectory for the UAV is still capacity-achieving in the more general setup remains as an open problem. How to extend the joint UAV trajectory and communication design to characterize the capacity region of other multiuser communication network models such as MAC and IFC with single-/multi-antenna nodes is also worthy of further investigation.


\end{itemize}
\appendices

\vspace{-0.3cm}
\section*{Appendix A: Proof of Lemma  \ref{lemm:P2}}  \label{apdx:1}
By relaxing the UAV maximum speed constraint \eqref{P1:speed}, one can show that problem (P2) satisfies the so-called time-sharing condition \cite{yu2006dual}, and therefore, strong duality holds between (P2) and its dual problem. Thus, we can solve (P2) by applying the Lagrange dual method.
Let  $\mu_k, k\in\{1,2\}$, denote the dual variable associated with the $k$th constraint in (10) in (P2). Then the partial Lagrangian of (P2) is given by
$\mathcal{L} (r,r_1,r_2,\{x(t), p_1(t), p_2(t)\}, \mu_1,\mu_2)=(1-\mu_1\alpha_1-\mu_2\alpha_1)r + \mu_1r_1 + \mu_2r_2$. To obtain the dual function under given $\mu_1$ and $\mu_2$,  we need to maximize the Lagrangian $\mathcal{L} (r,r_1,r_2,\{x(t), p_1(t), p_2(t)\}, \mu_1,\mu_2)$ by optimizing $r,r_1,r_2,\{x(t), p_1(t), p_2(t)\}$, subject to constraints \eqref{P1:11}, \eqref{P1:12}, and \eqref{P1:14}. To ensure the dual function unbounded from the above, it follows that $1-\mu_1\alpha_1-\mu_2\alpha_1=0$.  As a result, we only need to maximize $ \mu_1r_1 + \mu_2r_2$ by optimizing $r_1,r_2,\{x(t), p_1(t), p_2(t)\}$, subject to \eqref{P1:11}, \eqref{P1:12}, and \eqref{P1:14}. Towards this end, we consider the three cases with $\mu_1>\mu_2$, $\mu_1<\mu_2$, and $\mu_1=\mu_2$, respectively.


First, when $\mu_1>\mu_2$, we invoke the polymatroid structure \cite{tse1998multiaccess} for the above problem to remove constraint \eqref{P1:11}. Accordingly, we obtain the following equivalent problem as
\begin{align}
\max_{\{x(t),p_1(t),p_2(t)\}}&\frac{1}{T}\int_{0}^T\left(   \mu_1 \log_2(1+p_1(t) h_1(x(t))) +  \mu_2\log_2\left(1+ \frac{p_2(t)h_2(x(t))}{p_1(t)h_1(x(t)) +1}\right) \right)\mathrm{d}t \nonumber\\
\mathrm{s.t.}~~&\eqref{P1:12},~ \eqref{P1:14}.\label{P:50}
\end{align}
It then follows that problem \eqref{P:50} can be decoupled over any $t$ as the following subproblem, in which the index $t$ is omitted for brevity.
\begin{align}\label{vinfinite}
\max_{x,p_1,p_2} & (\mu_1-\mu_2)\log_2(1+p_1h_1(x)) +  \mu_2\log_2(1+p_1h_1(x)+p_2h_2(x))  \\
\mathrm{s.t.}~&p_1+p_2\le \bar P,\nonumber\\
~& p_1\geq 0, p_2\geq 0.\nonumber
\end{align}
Although problem \eqref{vinfinite} is still non-convex,  it can be shown that for $\mu_1>\mu_2$, $p_1\geq 0$, $p_2\geq 0$, $h_1(x)\leq h_1(-D/2)= \frac{\beta_0}{H^2}$, and  $h_2(x)\leq h_2(D/2)= \frac{\beta_0}{H^2}$, we have
\begin{align}
 (\mu_1-\mu_2)\log_2(1+p_1h_1(x))& \leq  (\mu_1-\mu_2)\log_2(1+ \bar P\frac{\beta_0}{H^2}),  \label{1}\\
\mu_2\log_2(1+p_1h_1(x)+p_2h_2(x)) &\leq   \mu_2\log_2(1+p_1\frac{\beta_0}{H^2}+p_2\frac{\beta_0}{H^2})  =   \mu_2\log_2(1+ \bar P\frac{\beta_0}{H^2}).  \label{2}
\end{align}
As the inequalities in \eqref{1} and \eqref{2} are active simultaneously only when $p^*_1=\bar P$, $p^*_2=0$, $x^*=-\frac{D}{2}$, it follows that $p^*_1=\bar P$, $p^*_2=0$, $x^*=-\frac{D}{2}$ are indeed the unique optimal solution to problem \eqref{vinfinite}. As a result, the optimal solution to the above problem is given as $p^*_1(t)=\bar P$, $p^*_2(t)=0$, $x^*(t)=-\frac{D}{2},\forall t$, which achieves zero rate for  GU 2 with $r_2 = 0$.

Next, when $\mu_2>\mu_1$, it can be similarly shown that the optimal solution to the above problem is $p^*_2(t)=\bar P$, $p^*_1(t)=0$, and $x^*(t)={D}/{2},\forall t$, which achieves zero rate for  GU 1 with $r_1 = 0$.

Furthermore, when  $\mu_1=\mu_2$, the above two trajectory and power allocation solutions are optimal for the above problem. As a result, the optimal solution is non-unique in this case, and the time-sharing between the two optimal solutions is required to achieve different rate pairs.

Note that $\frac{r_1}{r_2}=\frac{\alpha_1}{\alpha_2}$ must hold at the optimality of problem (P2). Therefore, it follows that $\mu_1=\mu_2$ must be true at the optimal dual solution of (P2). In this case, the solution in Lemma \ref{lemm:P2} is primal optimal to (P2). This thus completes the proof.
\vspace{-0.3cm}
\section*{Appendix B: Proof of Proposition  \ref{thm2}}  \label{apdx:2}
Suppose that the optimal UAV location solution to problem (P3) is $x^{*}\in[- {D}/{2}, 0)$, and the correspondingly obtained rate-pair $(r^*_1, r^*_2)$ satisfies $r^*_2\geq r^*_1$. Let $p_1$ and $p_2$ denote the corresponding UAV's transmit powers for GUs 1 and 2 to achieve the rate pair $(r^*_1, r^*_2)$. As $h_1(x^{*}) > h_2(x^{*})$,  GU 1 can use interference cancellation before decoding its own signal; thus, $r^*_1$ and $r^*_2$ can be explicitly expressed as $r^*_1= \log_2( 1+p_1h_1(x^{*}))$ and $r^*_2 = \log_2( 1+ \frac{ p_2h_2(x^{*})}{p_1h_2(x^{*})+1})$. In the following, we show that we can always find an alternative UAV location $\hat{x} = -x^*\in [0, {D}/{2}]$ to achieve a rate-pair $(\hat{r}_1, \hat{r}_2)$ with $(\hat{r}_1, \hat{r}_2)\succeq (r^*_1, r^*_2)$, by considering two cases with $r^*_2= r^*_1$ and $r^*_2> r^*_1$, respectively, as follows.

First, consider the case of $r^*_2= r^*_1$. It has been shown in Lemma \ref{lemma0} that for any rate-pair $(r^*_1, r^*_2)$, we can find a symmetric rate-pair $(r^*_2, r^*_1)$ at the location $\hat{x}=-x^{*}\in (0, {D}/{2}]$.  Thus,  it follows that $(\hat{r}_1, \hat{r}_2)=(r^*_2, r^*_1)= (r^*_1, r^*_2)$. Therefore, $\hat{x} \in (0, {D}/{2}]$ is also an optimal solution to problem (P3).

Next, consider the case of $r^*_2> r^*_1$. To start with, when the UAV is located at $\hat{x} = -x^*$, we denote the maximum rate of  GU 2 as $R_2=\log_2(1+{{\bar P}}h_2(\hat{x}))$, and the corresponding rate pair as $(0, R_2)$. Then, we construct an alternative rate-pair $(\hat{r}_1, \hat{r}_2)$ by time-sharing between $(0, R_2)$ and the symmetric rate-pair $(r^*_2, r^*_1)$. Let $\beta_1\ge 0$ and $\beta_2\ge 0$ denote the time-sharing ratios with $\beta_1+\beta_2=1$. Accordingly, the newly constructed rate-pair is expressed as $(\hat{r}_1,\hat{ r}_2) =\beta_1(r^*_2, r^*_1) + \beta_2(0, R_2)$. In particular, we choose $\beta_1 = \frac{R_2-r^*_2}{R_2-r^*_1}$ and $\beta_2 =1- \frac{R_2-r^*_2}{R_2-r^*_1}$ such that $\hat{r}_2=r^*_2$. Accordingly, we have $\hat{r}_1 =  \frac{R_2-r^*_2}{R_2-r^*_1}r^*_2$. As a result, in order to show $(\hat{r}_1, \hat{r}_2)\succeq (r^*_1, r^*_2)$ in this case, it remains to show that $\hat{r}_1 \ge r_1^*$, or equivalently,
\begin{align}\label{apdx:eq19}
\hat{r}_1-r^*_1 = \frac{R_2-r^*_2}{R_2-r^*_1}r^*_2-r^*_1 =\frac{r^*_2-r^*_1}{R_2-r^*_1}(R_2-r^*_1-r^*_2) \ge 0.
\end{align}
Notice that
\begin{align}\label{eq22}
r^*_1+r^*_2 &= \log_2( 1+p_1h_1(x^{*})) + \log_2( 1+ \frac{ p_2h_2(x^{*})}{p_1h_2(x^{*})+1}) \nonumber\\
&\overset{(a)}\leq \log_2( 1+p_1h_1(x^{*})) + \log_2( 1+ \frac{ p_2h_1(x^{*})}{p_1h_1(x^{*})+1}) \nonumber\\
&= \log_2( 1+\bar P h_1(x^{*}))\overset{(b)}= \log_2( 1+\bar P h_2(\hat x))=R_2,
\end{align}
where $(a)$ and $(b)$ hold due to the facts that $h_1(x^{*}) > h_2(x^{*})$ and $h_1(x^{*}) = h_2(\hat x)$, respectively. It thus follows that $R_2-r^*_1-r^*_2 \ge 0$. By using this together with $\frac{r^*_2-r^*_1}{R_2-r^*_1} \ge 0$, we have $\hat{r}_1-r^*_1 \ge 0$. Therefore, we have shown that $(\hat{r}_1, \hat{r}_2)\succeq (r^*_1, r^*_2)$.


By combining the above two cases, the proof is thus completed.

\vspace{-0.3cm}
\section*{Appendix C: Proof of Proposition \ref{trj_lemma}}  \label{apdx:B}
Proposition \ref{trj_lemma} is proved by contradiction. Suppose that the optimal rate-pair $(r^*_1, r^*_2)$  to problem (P1) is achieved by the UAV trajectory $\{x^*(t)\}$ and power allocation $\{{p}^*_k(t)\}$, in which there exists at least a location $x_A$ with  $x_A\in (x_{\rm I}, x_{\rm F})$, such that \eqref{eq:V1} is violated, i.e.,
 $\C_f(x_A) \nsubseteq {\rm{Conv}}(  \C_f(x_{\rm I}) \bigcup \C_f(x_{\rm F})\big)$.   In other words,  there exists at least one rate-pair $( r_1^{\Aaa},  r_2^{\Aaa})\in \C_f(x_A) $ (with the power allocation being denoted by $\{p^A_{k}\}$), such that $( r_1^{\Aaa},  r_2^{\Aaa})\notin {\rm{Conv}}\big(\C_f(x_{\rm I})\bigcup \C_f(x_{\rm F})\big)$. Then we show that we can always construct a new feasible UAV trajectory $\{\hat{x}(t)\}$ that achieves a larger objective value  for problem (P1) or equivalently a  rate-pair $(\hat{r}_1, \hat{r}_2)$ with  $(\hat{r}_1, \hat{r}_2)\succ (r^*_1, r^*_2)$. 
   \begin{figure}[!t]
\centering
\includegraphics[width=0.4\textwidth]{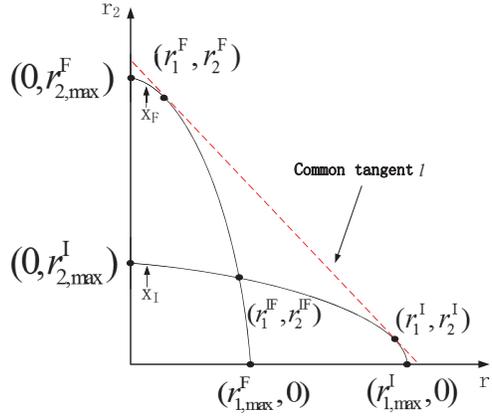}
\caption{Illustration of the common tangent of $\C_f(x_{\rm I})$ and $\C_f(x_{\rm F})$. } \label{tangent}\vspace{-0.6cm}
\end{figure}

First, we construct new UAV trajectory  $\{\hat{x}(t)\}$ and power allocation $\{\hat{p}_k(t)\}$ based on $\{x^*(t)\}$ and $\{{p}^*_k(t)\}$. To facilitate the design, we first partition the UAV flight/communication duration $T$ into three portions $[0,\beta_{\I}\Delta t ]$, $(\beta_{\I}\Delta t,T-\beta_{\F}\Delta t)$, and $[T-\beta_{\F}\Delta t, T ]$, in which $\beta_{\I} > 0$ and $\beta_{\F}> 0$ with $\beta_{\I}+\beta_{\F}=1$. We choose $\Delta t > 0$ to be sufficiently small such that $x^*(t)\approx x_{\rm I}, \forall t\in [0,\beta_{\I}\Delta t ]$, and $x^*(t)\approx x_{\rm F}, \forall t\in [T-\beta_{\F}\Delta t, T ]$. Accordingly,  from \eqref{P1:11} in (P1),  the average rate-pair of the two GUs  can be expressed as
\begin{align}\label{eq:adx556}
(r^*_1, r^*_2) = \frac{\Delta t}{T}\left( \beta_{\I}(r^{\I}_1,r^{\I}_2) + \beta_{\F}(r^{\F}_1,r^{\F}_2) \right) +  (R_1, R_2),
\end{align}
where $(r^{\I}_1,r^{\I}_2)$ and $(r^{\F}_1,r^{\F}_2)$ are the rate-pairs achieved at locations $x_{\rm I}$ and $x_{\rm F}$, respectively, and  $R_k= \frac{1 }{T}\int_{\beta_{\I}\Delta t}^{T-\beta_{\F}\Delta t} \log_2\big(1+ p^*_k(t)h_k(x^*(t))\big) {\rm{d}}t, k\in\{1,2\}$.  Here, note that $(r^{\I}_1,r^{\I}_2)$ and $(r^{\F}_1,r^{\F}_2)$ must be the rate-pairs in which $\C_f(x_{\I})$ and  $\C_f(x_{\F})$ touch with their common tangent, respectively, as shown in Fig. \ref{tangent}, since otherwise a larger objective value  for (P1) can be achieved by using rate-pairs $(r^{\I}_1,r^{\I}_2)$ and $(r^{\F}_1,r^{\F}_2)$ at $x_{\I}$ and $x_{\F}$.
Then, the new UAV trajectory  $\{\hat{x}(t)\}$ and power allocation $\{\hat{p}_k(t)\}$ are constructed by letting the UAV stay at the three locations $x_{\rm I}$, $x_{\rm F}$, and $x_{A}$ during the duration $\Delta t$, and using the trajectory $\{x^*(t)\}$ and power allocation $\{p^*_k(t)\}, t\in (\beta_{\I}\Delta t, T-\beta_{\F}\Delta t),$ during the remaining duration $T-\Delta t$. Let $\hat{\beta}_{\I}\Delta t$, $\hat{\beta}_{\F}\Delta t$, and $\hat{\beta}_{\Aaa}\Delta t$ denote the durations when the UAV stays at $x_{\rm I}$, $x_{\rm F}$, and $x_{A}$, respectively, where $\hat{\beta}_i\ge 0, i\in \{\I,\F, \Aaa\}$ are the corresponding time proportions with $\hat{\beta}_{\I}+\hat{\beta}_{\F}+\hat{\beta}_{\Aaa}=1$.

Accordingly, the average rate-pair achieved by the constructed solution $\{\hat{x}(t)\}$ and $\{\hat{p}_k(t)\}$  can be expressed as
\begin{align}\label{eq:adx557}
(\hat{r}_1, \hat{r}_2) = \frac{\Delta t}{T}\left(\hat{\beta}_{\I}( r^{\I}_1, r^{\I}_2)+ \hat{\beta}_{\F}(r^{\F}_1,r^{\F}_2)+ \hat{\beta}_{\Aaa}(r^{\Aaa}_1, r^{\Aaa}_2) \right) +  (R_1, R_2),
\end{align}
where $(r^{\Aaa}_1, r^{\Aaa}_2) \in \C_f(x_A)$ and $( r_1^{\Aaa}, r_2^{\Aaa})\notin {\rm{Conv}}\big(\C_f(x_{\rm I})\bigcup \C_f(x_{\rm F})\big)$. By comparing $(r^*_1, r^*_2)$ in \eqref{eq:adx556} and $(\hat{r}_1, \hat{r}_2)$ in \eqref{eq:adx557}, we have
\begin{align}\label{eq:apdx58}
\kern -0.7mm (\hat{r}_1, \hat{r}_2)- (r^*_1, r^*_2) = \frac{  \Delta t}{T} \left( \hat{\beta}_{\I}( r^{\I}_1, r^{\I}_2)+ \hat{\beta}_{\F}(r^{\F}_1,r^{\F}_2)+ \hat{\beta}_{\Aaa}(r^{\Aaa}_1, r^{\Aaa}_2)  - \beta_{\I}(r^{\I}_1,r^{\I}_2) - \beta_{\F}(r^{\F}_1,r^{\F}_2) \right).
\end{align}
Therefore, in order to show $(\hat{r}_1, \hat{r}_2)\succ (r^*_1, r^*_2)$, it remains to show that there exists a set of parameters $\hat{\beta}_{\I}$, $\hat{\beta}_{\F}$, $\hat{\beta}_{\Aaa}$, $\beta_{\I}$, and $\beta_{\F}$, such that
\begin{align}\label{eq:59:xj}
\hat{\beta}_{\I}( r^{\I}_1, r^{\I}_2)+ \hat{\beta}_{\F}(r^{\F}_1,r^{\F}_2)+ \hat{\beta}_{\Aaa}(r^{\Aaa}_1, r^{\Aaa}_2)  \succ \beta_{\I}(r^{\I}_1,r^{\I}_2) + \beta_{\F}(r^{\F}_1,r^{\F}_2).
\end{align}

Next, note that $(r^{\Aaa}_1, r^{\Aaa}_2)\notin {\rm{Conv}}\big(\C_f(x_{\rm I})\bigcup \C_f(x_{\rm F})\big)$ does not provide an explicit relation between  $(r^{\Aaa}_1, r^{\Aaa}_2)$ and $( r^{\I}_1, r^{\I}_2)$ and $(r^{\F}_1,r^{\F}_2)$, which thus makes it difficult to verify the inequality in \eqref{eq:59:xj}. To overcome this difficulty, we introduce an equivalent interpretation of $(r^{\Aaa}_1, r^{\Aaa}_2)\notin {\rm{Conv}}\big(\C_f(x_{\rm I})\bigcup \C_f(x_{\rm F})\big)$ in the following lemma, which is based on the properties of the capacity region $\C_f(x)$.
\begin{figure}[!t]\label{example}
\centering
\subfigure[]{\includegraphics[width=2.1in, height=1.7in]{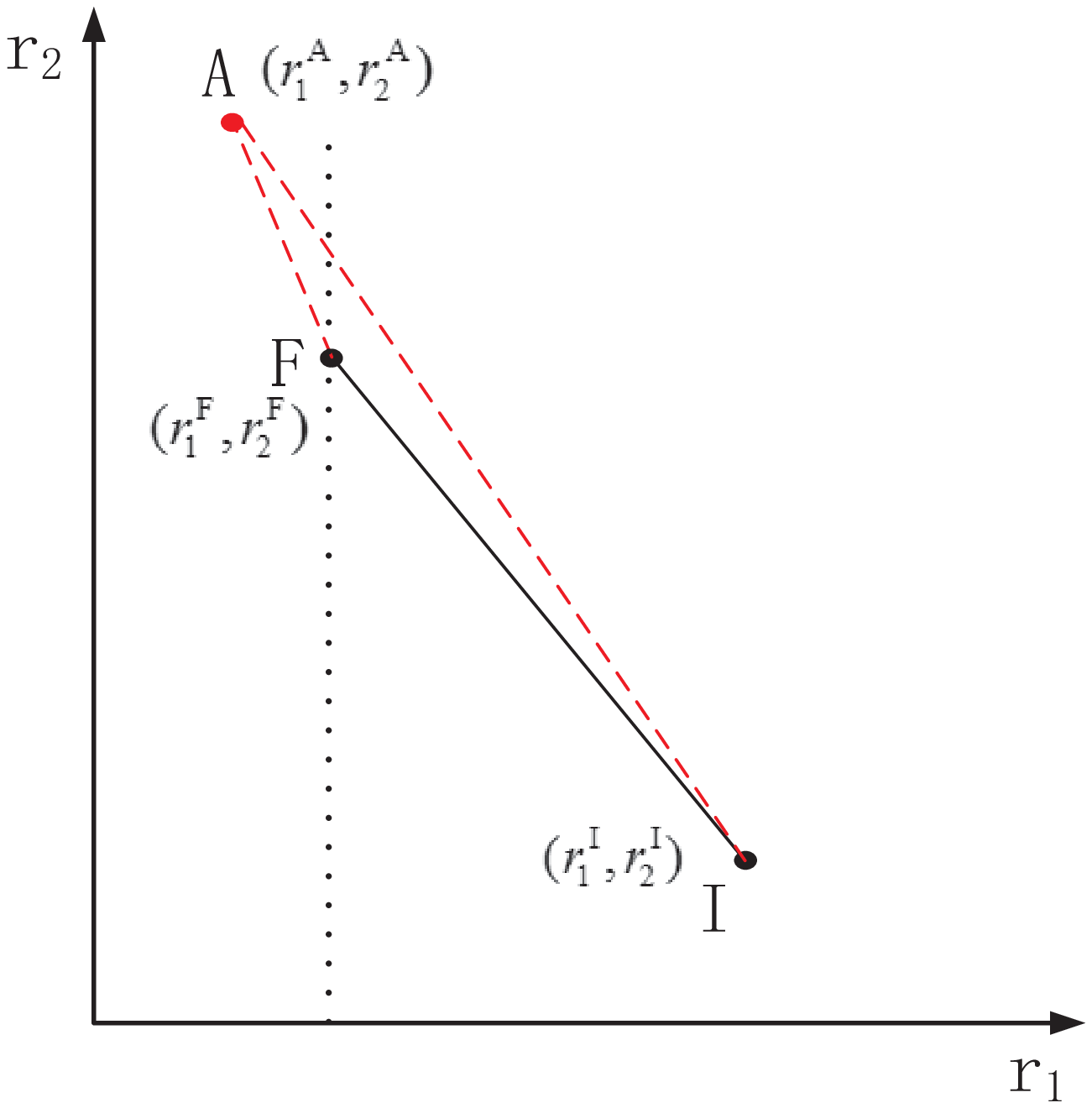}}
\subfigure[]{\includegraphics[width=2.1in, height=1.7in]{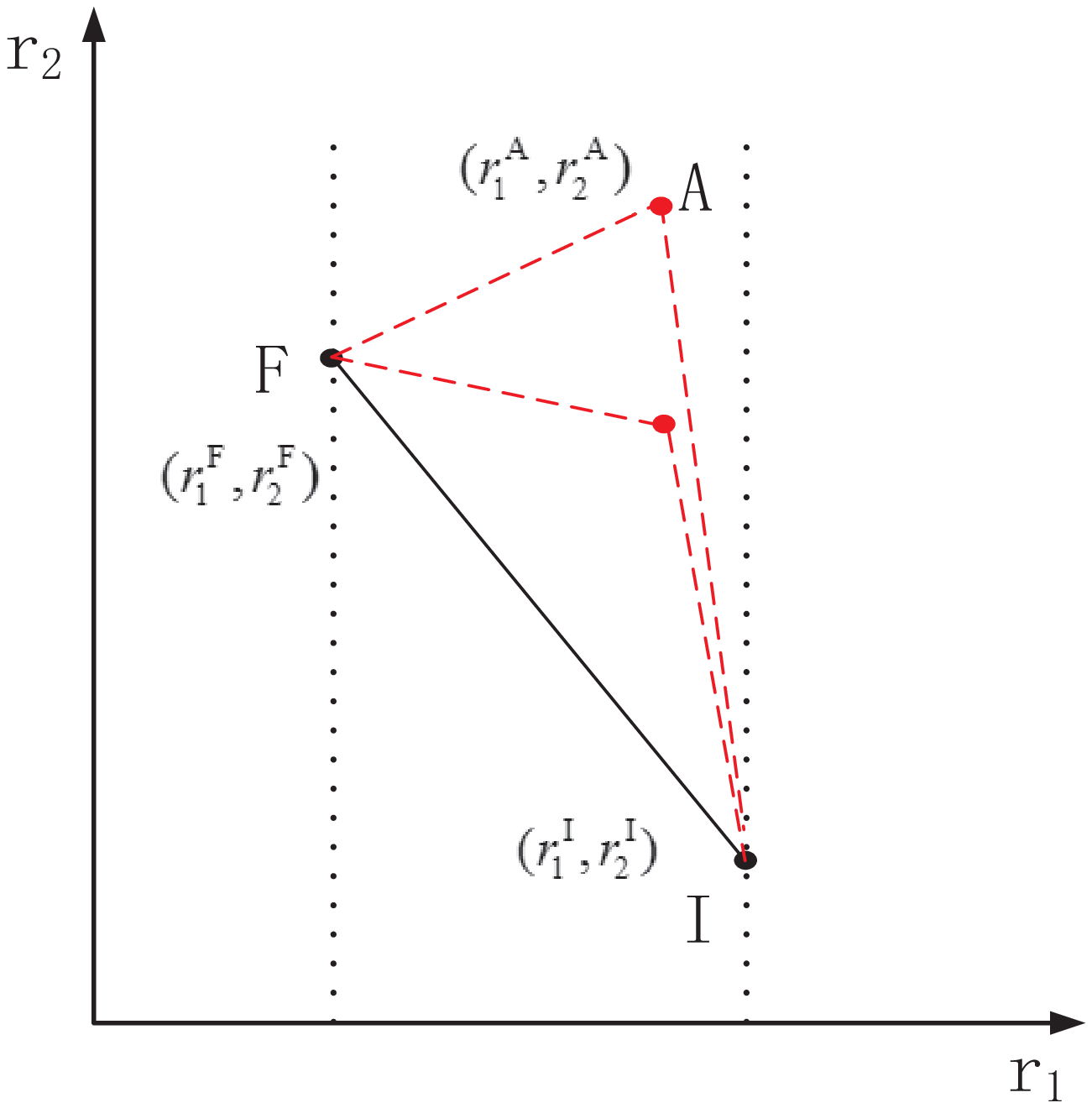}}
\subfigure[]{\includegraphics[width=2.1in, height=1.7in]{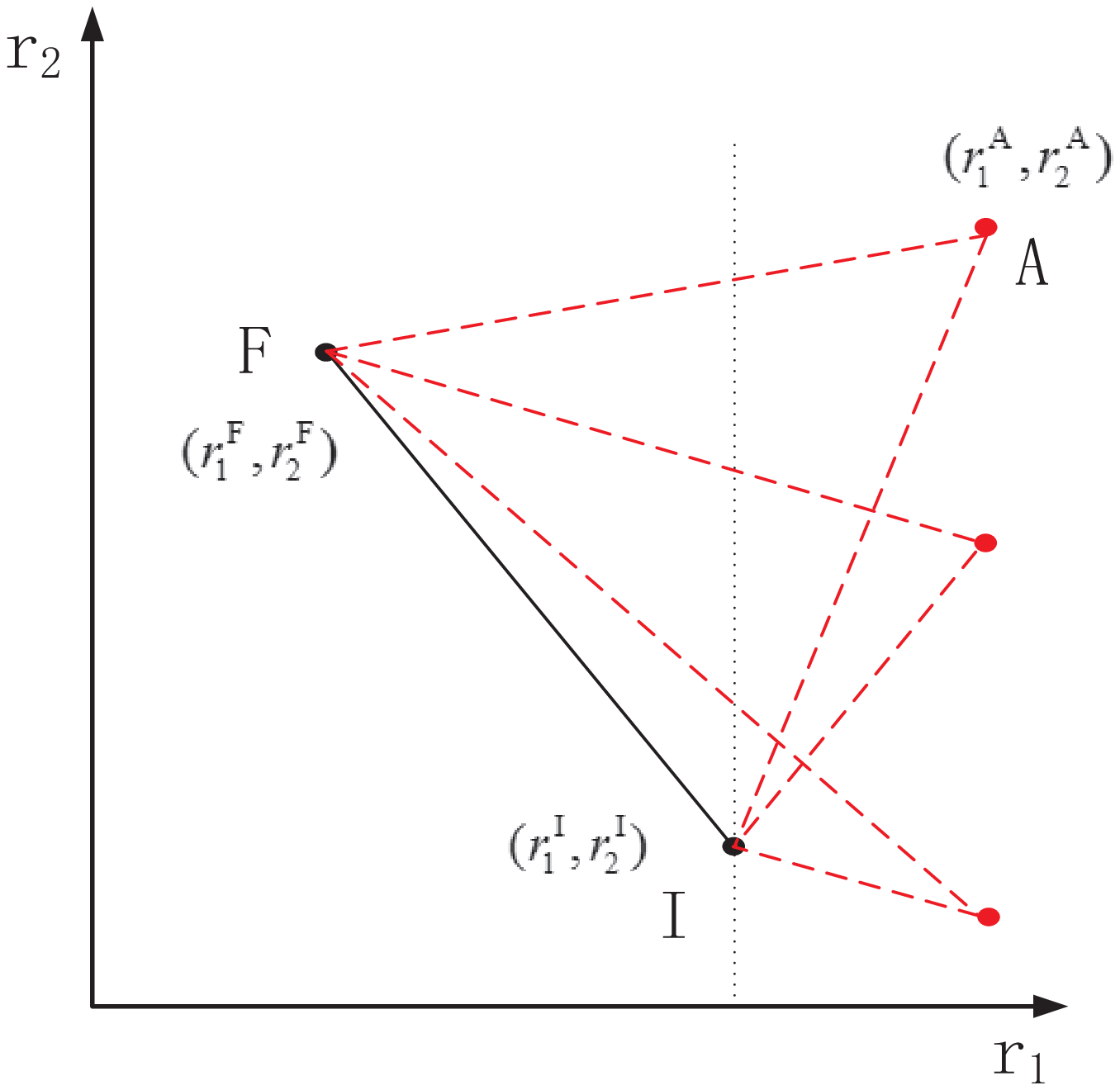}}
\caption{The three cases where  point  $(r_1^{\Aaa},  r_2^{\Aaa})$  does not lie inside the triangle region, $\C_{\rm IF}$, defined in \eqref{equiva_triangle}:  a) $0 < r_1^{\Aaa} \leq   r_1^{\F}$,   b) $r_1^{\F} <  r_1^{\Aaa} <  r_1^{\I} $, and c) $r_1^{\I} \leq r_1^{\Aaa} <  r_{1,\max}^{\I}$.} \label{UAV_location}\label{example}\vspace{-0.8cm}
\end{figure}
\begin{lemma}\label{lem_convhull}
For any location $x\in[x_{\rm I}, x_{\rm F}]$, $\C_f(x) \subseteq {\rm{Conv}}\big(\C_f(x_{\rm I})\bigcup \C_f(x_{\rm F})\big)$ \emph{if} and \emph{only if}  $\C_f(x) \subseteq \C_{\rm IF}$ where $\C_{\rm IF}$ is a  triangle region given by
\begin{align}\label{equiva_triangle}
\C_{\rm IF} = \bigg\{( r_1, r_2):  r_2 \le   k_{\rm IF} (r_1 -  r^{\I}_1 ) +  r^{\I}_2, r_1\geq 0, r_2\geq 0, k_{\rm IF} = \frac{ r^{\I}_2- r^{\F}_2}{r^{\I}_1- r^{\F}_1} \bigg\}.
\end{align}
\end{lemma}
\begin{proof}
Please refer to Appendix F.
\end{proof}
Based on  Lemma \ref{lem_convhull},    $(r^{\Aaa}_1, r^{\Aaa}_2)\notin {\rm{Conv}}\big(\C_f(x_{\rm I})\bigcup \C_f(x_{\rm F})\big)$ implies    $( r_1^{\Aaa},  r_2^{\Aaa})  \notin \C_{\rm IF}$, i.e., 
\begin{align}\label{eq27}
 r_2^{\Aaa}>  k_{\rm IF} (r^{\Aaa}_1 -  r^{\I}_1 ) +  r^{\I}_2.
\end{align}

Finally, with \eqref{eq27} at hand, we find parameters $\hat{\beta}_{\I}$, $\hat{\beta}_{\F}$, $\hat{\beta}_{\Aaa}$, $\beta_{\I}$, and $\beta_{\F}$, to ensure \eqref{eq:59:xj}.
Note that for $x_A\in (x_{\rm I}, x_{\rm F})$, we have  $0 < r_1^{\Aaa}<  r_{1, \max}^{\I}$ where $r^{\I}_{1,\max}$ is the maximum achievable rate of  GU 1 when the UAV is placed at location $x_{\I}$ and allocates $\bar P$ to  GU 1, i.e., $r^{\I}_{1,\max} = \log_2(1+\frac{\bar P\beta_0}{ (x_{\rm I}+\frac{D}{2})^2+H^2 })$. In general, there are overall three possible cases depending on the values of $r_1^{\Aaa}$, i.e.,  1) $0 < r_1^{\Aaa} \leq   r_1^{\F}$,   2) $r_1^{\F} <  r_1^{\Aaa} <  r_1^{\I} $, and 3) $r_1^{\I} \leq r_1^{\Aaa} <  r_{1,\max}^{\I}$, for which the rate-pair $(r_1^{\Aaa},r_2^{\Aaa})$ is shown in Fig. \ref{example} based on Lemma \ref{lem:inequality} in Appendix F.  For simplicity, we only discuss case 1), as  cases 2) and 3) can be similarly analyzed.
When  $0 < r_1^{\Aaa} \leq  r_1^{\F} $, it follows that $r_1^{\Aaa} \leq  r_1^{\F}<  r_1^{\I}$.
  From \eqref{eq27}, we have $k_{\rm AF}  = \frac{ r_2^{\Aaa}- r_2^{\I}}{ r_1^{\Aaa} - r^{\I}_1}<k_{\rm IF} <0$ where $k_{\rm AI}$ is the slope of the line across points $( r_1^{\Aaa}, r_2^{\Aaa})$ and  $( r_1^{\I}, r_2^{\I})$,  and  the line segment  is given by
  \begin{align}\label{eq29}
r_2= k_{\rm AI}  (r_1 -  r^{\I}_1 ) +  r^{\I}_2,
  \end{align}
with $r_1 \in [ r_1^{\Aaa},  r^{\I}_1 ]$.
Note that performing proper time-sharing (convex combination) between points $( r_1^{\Aaa}, r_2^{\Aaa})$ and   $( r_1^{\I}, r_2^{\I})$, i.e., $\hat{\beta}_{\F}=0$,  can achieve any point $(r_1, r_2)$ on line segment in \eqref{eq29} that is above the line segment $r_2 =   k_{\rm IF} (r_1 -  r^{\I}_1 ) +  r^{\I}_2$, $r_1\in (r^{\F}_1, r^{\I}_1)$.
 In other words,
we can always find a convex combination of $( r^{\I}_1, r^{\I}_2)$, $(r^{\F}_1,r^{\F}_2)$, and $(r^{\Aaa}_1, r^{\Aaa}_2)$, which strictly outperforms any given convex combination of $( r^{\I}_1, r^{\I}_2)$ and $(r^{\F}_1,r^{\F}_2)$ (with $\beta_{\I} >0$ and $\beta_{\F} > 0$).
  Therefore, there exist parameters $\hat{\beta}_{\I}$, $\hat{\beta}_{\F}$, $\hat{\beta}_{\Aaa}$, $\beta_{\I}$, and $\beta_{\F}$, to ensure \eqref{eq:59:xj}, which thus completes the proof.


\vspace{-0.3cm}
\section*{Appendix D: Proof of Theorem \ref{prop4}}  \label{apdx:prop4}
Notice that when $x_{\rm I} = x_{\rm F}$, Theorem \ref{prop4} follows directly. Therefore, we only need to focus on the proof in the case with  $x_{\rm I}<x_{\rm F}$, by contradiction.

Suppose that there exists an optimal UAV trajectory solution to problem (P1), in which the UAV flies at the  speed $V^{\star}$ less than $V$ ($0<V^{\star}<V$)\footnote{Note that if the UAV hovers at some locations between $[x_A-\delta_d/2, x_A+\delta_d/2]$, we can obtain its average speed $V^*$ during this interval without loss of optimality, which is always larger than zero. } over an infinitesimal interval $[x_A-\delta_d/2, x_A+\delta_d/2]$ containing location $x_A$ ($x_{\rm I}<x_A-\delta_d/2<x_A+ \delta_d/2<x_{\rm F}$) with  $\delta_d>0$ being infinitesimal.  Then  the time needed for flying over this interval is $t=\frac{\delta_d}{{V}^{\star}}$ and the UAV's location can be assumed constant within this interval as it is sufficiently small.
As a result, we can always construct an alternative feasible trajectory which is same as the assumed trajectory except that we reallocate the flying and hovering time at locations $x_A$, and $x_{\rm I}$ and $x_{\rm F}$. Specifically,  we let the UAV fly at the maximum speed $V$ over $[x_A-\delta_d/2, x_A+\delta_d/2]$ and the time saved due to maximum speed flying is given by $\triangle t = \delta_d\left(\frac{1}{V^{\star}}-\frac{1}{{V}}\right)>0$. Then, we let the UAV perform proper hovering time-sharing between $x_{\rm I}$ and $x_{\rm F}$ by using the saved time $\triangle t$.
From Proposition  \ref{trj_lemma}, we know that for the optimal trajectory solution to problem (P1),  $\C_f(x_A) \subseteq {\rm{Conv}}\big(\C_f(x_{\rm I})\bigcup \C_f(x_{\rm F})\big)$ for $x_{\rm I}<x_A<x_{\rm F}$.
Thus it can be shown that the newly constructed UAV trajectory together with optimized power allocation can always achieve a rate-pair that is  componentwise no smaller than the assumed one, which thus completes the proof.
\vspace{-0.3cm}
\section*{Appendix E: Proof of Theorem \ref{prop5}}  \label{apdx:prop5}
Based on Lemma \ref{lemma0}, we only need to consider $r_2\geq r_1$ for problem (P1). To obtain the optimal solution under the high SNR assumption, we first obtain an upper bound of the optimal objective value of (P1) and then we show that this upper bound is tight and can be achieved by $x^*(t)=D/2, t\in \T$.

From constraints \eqref{P1:11} or \eqref{MAC:ineq3} in (P1), it follows that $r=r_1+r_2\le \frac{1}{T}\int_{0}^T\log_2\big(1+ p_1(t)h_1(x(t)) +p_2(t)h_2(x(t))\big){\rm{d}}t\overset{(a)}{\le} \frac{1}{T}\int_{0}^T\log_2(1+ \frac{\bar P \beta_0}{H^2}){\rm{d}}t \overset{(b)}{\thickapprox} \log_2( \frac{\bar P \beta_0}{H^2})$, where $(a)$ holds due to $h_1(x(t))\le \frac{\beta_0}{H^2} $, $h_2(x(t))\le \frac{\beta_0}{H^2}$, and $ p_1(t)+p_2(t)\le \bar P$, and $(b)$ holds due to the high SNR assumption $\frac{\bar P \beta_0}{ H^2 }\ge \frac{\bar P \beta_0}{ D^2 + H^2 }\gg1$. Thus, in this case, the optimal objective value of (P1) is upper-bounded by $\log_2( \frac{\bar P \beta_0}{H^2})$.

Next, we show that the above upper bound can be achieved by a feasible solution to problem (P1) with $x^*(t)= x^*, \forall\, t\in \T$, in which case problem (P1) reduces to (P3). From Proposition  \ref{thm2}, we know that  GU 2  performs interference cancellation before decoding its own signal when $r_2\geq r_1$. Furthermore, based on the assumption $\frac{\bar P \beta_0}{ D^2 + H^2 }\gg1$ and  $r_2\geq r_1$, it can be shown $\frac{p_2\beta_0}{ (x+\frac{D}{2})^2+H^2 } \gg1$. As a result,  the achievable rates of GUs 1 and 2 under the high SNR assumption  can be expressed as 
{\small\begin{align}
r_1&=\log_2\left(1+\frac{\frac{p_1\beta_0}{ (x+\frac{D}{2})^2+H^2 }}{\frac{p_2\beta_0}{ (x+\frac{D}{2})^2+H^2 } + 1}\right)\thickapprox \log_2\left( 1+\frac{p_1}{p_2}\right), \label{eq:58}\\
r_2&=\log_2\left(1+\frac{p_2\beta_0}{ (x-\frac{D}{2})^2+H^2 }\right)\thickapprox \log_2\left( \frac{p_2\beta_0}{ (x-\frac{D}{2})^2+H^2 }\right).\label{eq:59}
\end{align}}From \eqref{eq:58} and \eqref{eq:59}, it can be observed that $r_1$ is independent of $x$ and $r_2$ increases monotonically with $x$ for $x\in [-D/2, D/2]$. Thus, the  objective function of (P3) is $r=r_1+r_2=\log_2( \frac{\bar P \beta_0}{ (x-\frac{D}{2})^2+ H^2 })$ where the optimal UAV location is $x^*=D/2$ with the maximum objective value  $\log_2( \frac{\bar P \beta_0}{ H^2 })$. Based on the above results, the capacity region $\C_{\rm h-SNR}(V, T, \bar P)$ can be obtained as in \eqref{eq:highSNR}. This thus completes the proof.

\vspace{-0.3cm}
\section*{Appendix F:  Proof of Lemma \ref{lem_convhull}}  \label{apdx:B}
To start with, we provide Lemmas \ref{lem5} and \ref{lem:inequality} below to facilitate the proof of Lemma \ref{lem_convhull} in the next. First, we consider two UAV locations $x_{\rm B}$ and $x_{\rm C}$, with $-\frac{D}{2}\leq x_{\rm C}<x_{\rm B}\leq \frac{D}{2}$. For any boundary point $(r_1, r_2)\in \partial \C_f(x_{m}), m \in\{\rm B, C\}$, we use a rate function $r_2 = r_2^m(r_1)$ to express their relation. Also, let ${r}^{m }_{k, \max}\triangleq \log_2\left(1+\bar Ph^m_k\right)$ denote the maximum rate of GU $k$ when the UAV is located at $x_m$, where $h^m_k =  \frac{\beta_0}{ (x_{m}-x_k)^2+H^2 }$, $m\in\{\rm B, C\}$, $k\in\{1,2\}$. We then have Lemma \ref{lem5} as follows to reveal an interesting property of $r_2^B(r_1)$ and  $r_2^C(r_1)$.

%
%
%


\begin{lemma}\label{lem5}
$r_2^B(r_1)$ and  $r_2^C(r_1)$ have one unique intersecting point, denoted by $(\bar{r}^{BC}_1,\bar{r}^{BC}_2)$, where $0<\bar{r}^{BC}_1<{r}^{B}_{1, \max}$ and $0<\bar{r}^{BC}_2<{r}^{C}_{2, \max}$. Furthermore, when $0\leq r_1 <\bar{r}^{BC}_1$, it follows that $r_2^B(r_1)>r_2^C(r_1) $; otherwise when  $\bar{r}^{BC}_1< r_1\leq  {r}^{B }_{1, \max}$, $r_2^B(r_1)<r_2^C(r_1)$ must hold.
\end{lemma}
\begin{proof}
Since $-\frac{D}{2}\leq x_{\rm C}<x_{\rm B}\leq \frac{D}{2}$, we must have ${r}^{B }_{1, \max}<{r}^{C}_{1, \max}$ and  ${r}^{B }_{2, \max}>{r}^{C}_{2, \max}$. Thus,  it follows that $0<\bar{r}^{BC}_1<{r}^{B }_{1, \max}$ and $0<\bar{r}^{BC}_2<{r}^{C}_{2, \max}$.

Next, we prove that the intersecting point $(\bar{r}^{BC}_1,\bar{r}^{BC}_2)$ is unique, by considering three possible cases, depending on the values of  $x_{\rm B}$  and $x_{\rm C}$, i.e., 1) $ 0 \leq x_{\rm C} < x_{\rm B}\leq \frac{D}{2}$; 2) $-\frac{D}{2}\leq x_{\rm C}<0<x_{\rm B}\leq \frac{D}{2}$; and 3) $-\frac{D}{2}\leq x_{\rm C}<x_{\rm B}\leq 0$, respectively. Since case 3) is similar to case 1), we only  analyze the first two cases in the following for brevity.

In case 1), we have the following inequality regarding the channel gains of the two GUs with the UAV locations being $x_{\rm B}$ and $x_{\rm C}$: $h^B_2>  h^C_2> h^C_1>h^B_1>0$. At both locations $x_{\rm B}$ and $x_{\rm C}$,  GU 2 will decode  GU 1's signal first and then perform interference cancellation. As such, the achievable rates of GUs 1 and 2  can be expressed as $r_1^m = \log_2 \left(1+ \frac{p_1^mh_1^m}{p_2^mh_1^m+1} \right)$ and  $r_2^m = \log_2(1+ p_2^mh_2^m)$, respectively, $m\in\{\rm B,C\}$.
By eliminating $p_1^m$ and $p_2^m$ with $p_1^m+p_2^m =\bar P$, it yields
$r_2^m= \log_2 (1+\frac{(\bar P h_1^m-2^{r_1^m}+1)h_2^m}{2^{r_1^m}h_1^m}), m\in\{\rm B,C\}$.
Note that the intersecting point  $(\bar{r}^{BC}_1,\bar{r}^{BC}_2 )$ should satisfy both equalities above. Therefore, we have
{
\begin{align}\label{eq:33}
\frac{(\bar P h^B_1-2^{\bar{r}^{BC}_1}+1)h^B_2}{2^{\bar{r}^{BC}_1}h^B_1 }  =\frac{(\bar P h^C_1-2^{\bar{r}^{BC}_1}+1)h^C_2}{2^{\bar{r}^{BC}_1}h^C_1 }.
\end{align}}
From \eqref{eq:33}, we obtain $\bar{r}^{BC}_1 = \log_2\left(\bar Ph^B_1h^C_1\left(\frac{h^B_2-h^C_2}{h^B_2h^C_1-h^B_1h^C_2}\right) +1\right)$. Since $\bar Ph^B_1h^C_1\left(\frac{h^B_2-h^C_2}{h^B_2h^C_1-h^B_1h^C_2}\right) >0$, there always exists a unique solution of $\bar{r}^{BC}_1>0$.

In case 2), we have the following inequalities: $h^B_2> h^B_1>0$, $ h^C_1> h^C_2>0$, $h^B_2> h^C_2>0$, and $ h^C_1> h^B_1>0$. Therefore, when the UAV is located at $x_{\rm B}$, the achievable rates of GUs 1 and 2 are given as those in case 1). When the UAV is located at $x_{\rm C}$,   GU 1 needs to decode  GU 2's signal first and then perform interference cancellation  where  the achievable rates of GUs 1 and 2 can be written as $r_1^C = \log_2 \left(1+  p_1^Ch_1^C\right)$ and $r_2^C = \log_2\left(1+ \frac{p_2^Ch_2^C}{p_1^Ch_1^C+1}\right)$, respectively.
Then, for the intersecting point  $(\bar{r}^{BC}_1,\bar{r}^{BC}_2 )$, we have
{
\begin{align}\label{eq38}
\frac{(\bar P h^C_2+1) h^C_1}{(2^{\bar{r}^{BC}_1}-1) h^C_2+ h^C_1}-\frac{(\bar P h^B_1+1) h^B_2}{2^{\bar{r}^{BC}_1}h^B_1}+\frac{h_2^B}{h_1^B}  -1=0.
\end{align}}
Let $w\triangleq  2^{\bar{r}^{BC}_1}$. Then, \eqref{eq38} can be equivalently transformed to
\begin{small}
\begin{align}\label{eq39}
\frac{ h^C_2}{ h^C_1}\left(\frac{ h^B_2}{ h^B_1} -1\right)w^2 -\left( \bar Ph_2^C\left(1-\frac{h_2^B}{h_1^C}\right) - \frac{2h_2^Bh_2^C}{h_1^Bh_1^C} +\frac{h_2^B}{h_1^B} +\frac{h_2^C}{h_2^C} \right)w + h_2^{B}\left( \frac{h_2^{C}}{h_1^{C}}-1 \right)(\bar P +h_1^{B})=0,
\end{align}
\end{small}
\kern -1.6mm where the LHS of \eqref{eq39} is a quadratic function with respect to $w$. Based on inequalities of the channel gains, it can be shown that \eqref{eq39} always has a unique  root $w>1$, i.e., $\bar{r}^{BC}_1>0$ is unique.

Finally, by using the result that $\bar{r}^{BC}_1>0$ is unique, together with the facts that ${r}^{B }_{1, \max} < {r}^{C}_{1, \max}$ and ${r}^{B }_{2, \max} > {r}^{C}_{2, \max}$, it is evident that when $0\leq r_1 <\bar{r}^{BC}_1$, it follows that $r_2^B(r_1)>r_2^C(r_1) $; otherwise when  $\bar{r}^{BC}_1< r_1\leq  {r}^{B }_{1, \max}$, $r_2^B(r_1)<r_2^C(r_1)$ must hold. Therefore, the second part of this lemma is proved. Hence, Lemma \ref{lem5} follows.
\end{proof}


Next, we denote the upper right common tangent of convex hull of the union of $ \C_f(x_{\rm B})$ and
$\C_f(x_{\rm C})$ by $\mathcal{\ell}$ and it touches two points on the boundaries of  $\C_f(x_{\rm B})\bigcup \C_f(x_{\rm C})$, denoted by $(\bar r^B_1, \bar r^B_2)$ and $(\bar r^C_1, \bar r^C_2)$, respectively. Then we have the following lemma.
\begin{lemma}\label{lem:inequality}
It follows that $0 \leq \bar r^B_1<\bar{r}^{BC}_1<\bar r^C_1$ and $0 \leq \bar r^C_2<\bar{r}^{BC}_2<\bar r^B_2$.
\end{lemma}
\begin{proof}
Since both $r_2^B(r_1)$ and  $r_2^C(r_1)$  are monotonically decreasing functions with respect to $r_1$, $0 \leq \bar r^C_2<\bar{r}^{BC}_2<\bar r^B_2$ holds \emph{if} and \emph{only if} $0 \leq \bar r^B_1<\bar{r}^{BC}_1<\bar r^C_1$ holds. Thus, it only remains to show $0 \leq \bar r^B_1<\bar{r}^{BC}_1<\bar r^C_1$. In the following, we focus on proving   $\bar{r}^{BC}_1<\bar r^C_1$  by contradiction, while  $\bar r^B_1<\bar{r}^{BC}_1$ can be similarly verified and thus is omitted.

First, suppose that $\bar r^C_1< \bar{r}^{BC}_1$. From Lemma \ref{lem5}, it follows that $r_2^B(\bar r^C_1)>r_2^C(\bar r^C_1) $ and hence  $ (\bar r^C_1, \bar r^C_2) \prec(\bar r^C_1, \bar r^B_2)$. Since $(\bar r^C_1, \bar r^B_2) \in  \C_f(x_{\rm B})$,   we must have $(\bar r^C_1, \bar r^C_2) \in  \C_f(x_{\rm B})$ and $(\bar r^C_1, \bar r^C_2)$ is not on the boundary of $\C_f(x_{\rm B})$, i.e., $(\bar r^C_1, \bar r^C_2) \in  {\rm {int}}( \C_f(x_{\rm B}))$. This contradicts that line $\ell$  is the common tangent of $\C_f(x_{\rm B})$ and $\C_f(x_{\rm C})$ because it crosses an interior point of $\C_f(x_{\rm B})$.

Next, suppose that $\bar r^C_1= \bar{r}^{BC}_1$. Then we have $(\bar r^C_1, \bar r^C_2) =(\bar{r}^{BC}_1,\bar{r}^{BC}_2 )\in  \C_f(x_{\rm B})$ and $(\bar r^C_1, \bar r^C_2)$ is also on the boundary of $\C_f(x_{\rm B})$ as $(\bar r^B_1, \bar r^B_2)$, i.e., $(\bar r^C_1, \bar r^C_2)\in \partial \C_f(x_{\rm B})$. Note that if $x_{\rm B}\neq 0$, then $\C_f(x_{\rm B})$ is a strictly convex set, in which case line $\ell$ cannot be the common tangent of $\C_f(x_{\rm B})$ because it crosses two points $(\bar r^C_1, \bar r^C_2)$ and $(\bar r^B_1, \bar r^B_2)$ of a strictly convex boundary. By contrast, if $x_{\rm B}=0$, then $\C_f(x_{\rm B})$ is a convex but not strictly convex set, and more specifically, it is an  equilateral triangle region, with its boundary being a line segment that lies on line $\ell$. Since $x_{\rm B}\neq x_{\rm C}$, it follows from Lemma \ref{lem5} that  line $\ell$ intersects with the boundary of $\C_f(x_{\rm C})$, which thus contradicts that line $\ell$ is the tangent of $\C_f(x_{\rm C})$.

By combining the two cases, the proof is completed.
\end{proof}

Now, with Lemmas \ref{lem5} and \ref{lem:inequality} obtained, we are ready to prove Lemma 5.  We first prove the ``only if" part. By the definition of a common tangent (supporting hyperplane theorem \cite{Boyd}), all the achievable rate-pairs in the convex hull should lie in the same half (lower left) plane separated by the tangent, i.e., ${\rm{Conv}}\big(\C_f(x_{\rm I})\bigcup \C_f(x_{\rm F})\big) \subseteq \C_{\rm IF}$.
 Thus, $\C_f(x_A) \subseteq {\rm{Conv}}\big(\C_f(x_{\rm I})\bigcup \C_f(x_{\rm F})\big)\Rightarrow \C_f(x_{A}) \subseteq \C_{\rm IF}$.

Next, we prove the ``if" part, i.e.,  $\C_f(x_A) \subseteq \C_{\rm IF}  \Rightarrow \C_f(x_A) \subseteq {\rm{Conv}}\big(\C_f(x_{\rm I})\bigcup \C_f(x_{\rm F})\big)$. This is proved by contradiction. Suppose that $\C_f(x_A) \subseteq \C_{\rm IF}$ but $\C_f(x_A) \nsubseteq  {\rm{Conv}}\big(\C_f(x_{\rm I})\bigcup \C_f(x_{\rm F})\big)$. It means that there exists at least a rate-pair that satisfies $(\bar r^A_1, \bar r^A_2)\in \C_f(x_A)$ but $(\bar r^A_1, \bar r^A_2) \in  \C_{\rm IF}  \backslash  {\rm{Conv}}\big(\C_f(x_{\rm F})\bigcup \C_f(x_{\rm I})\big)$. Note that $(\bar r^{\I}_1, \bar r^{\I}_2)$ and $(\bar r^{\F}_1, \bar r^{\F}_2) \in {\rm{Conv}}\big(\C_f(x_{\rm I})\bigcup \C_f(x_{\rm F})\big)$. From Lemma \ref{lem:inequality}, we have $0 \leq \bar r^{\F}_1<\bar{r}^{\rm IF}_1<\bar r^I_1$.
Thus, the rate-pairs that lie on and below the line segment between $(\bar r^{\rm F}_1, \bar r^{\rm F}_2)$ and $(\bar r^{\rm I}_1, \bar r^{\rm I}_2)$ also lie within  ${\rm{Conv}}\big(\C_f(x_{\rm I})\bigcup \C_f(x_{\rm F})\big)$, i.e., $(r_1, r_2)\in {\rm{Conv}}\big(\C_f(x_{\rm I})\bigcup \C_f(x_{\rm F})\big)$, where $\bar r^{\rm F}_1 \leq r_1 \leq \bar r^{\rm I}_1$ and $(r_1, r_2)\in \C_{\rm IF}$. Then, the rate-pair $(\bar r^A_1, \bar r^A_2)$ must satisfy one of the following two cases: 1) $0<\bar r^A_1< \bar r^{\rm F}_1$, and 2) $\bar r^{\rm I}_1<\bar r^{A}_1$.
  If $0<\bar r^A_1< \bar r^{\rm F}_1$, since $(\bar r^A_1, \bar r^A_2)\in \C_f(x_A)$ and $(\bar r^A_1, \bar r^A_2)\notin \C_f(x_{\rm F})$, it follows that $(\bar r^A_1, \bar r^A_2)\succ (\bar r^A_1, r^{\rm F}_2(\bar r^A_1))$. This implies that the intersecting point, denoted by $(r_1^{FA}, r_2^{FA})$, satisfies $r_1^{FA}<\bar r^A_1$. From Lemma \ref{lem5},  we have $(\bar r^{\rm F}_1, \bar r^{\rm F}_2)\prec (\bar r^{\rm F}_1, r^{A}_2(\bar r^{\rm F}_1))$ due to $ r_1^{FA}< \bar r^{A}_1< \bar r^{\rm F}_1$. It thus suggests that there exists a rate-pair $(\bar r^{\rm F}_1, r^{A}_2(\bar r^{\rm F}_1))\in  \C_f(x_{\rm F})$ but $(\bar r^{\rm F}_1, r^A_2(\bar r^{\rm F}_1))\notin  \C_{\rm IF}$, which contradicts the assumption $\C_f(x_{\rm F}) \subseteq \C_{\rm IF}$. The case of $\bar r^{\rm I}_1<\bar r^A_1$ can be analyzed similarly, where the result also contradicts the assumption $\C_f(x_{\rm F}) \subseteq \C_{\rm IF}$. This thus completes the proof of Lemma \ref{lem_convhull}.

\vspace{-0.5cm}
\bibliographystyle{IEEEtran}
\bibliography{IEEEabrv,mybib}

\end{document}